\documentclass[10pt,journal,twoside]{IEEEtran}
\usepackage{generic}
\usepackage{epsfig,color}
\usepackage{amsthm}
\usepackage{amsmath}    
\usepackage{bm}
\usepackage{epstopdf}
\usepackage{widetext}
\usepackage{threeparttable}
\usepackage{amssymb}
\usepackage{url}
\usepackage{enumitem}
\usepackage{multirow}
\usepackage{hhline}
\usepackage{booktabs}

\usepackage[linesnumbered,boxed,commentsnumbered,ruled,vlined,longend]{algorithm2e}
\SetAlgorithmName{Procedure}
\usepackage{comment}
\setlist[itemize]{leftmargin=*}

\makeatother
\DeclareMathAlphabet\mathbfcal{OMS}{cmsy}{b}{n}

\usepackage{float}
\usepackage[caption = false]{subfig}
\usepackage{caption}

\newtheorem{mylem}{Lemma}

\newtheorem{asmp}{Assumption}
\newtheorem{mycor}{Corollary}
\newtheorem{myprs}{Proposition}



\usepackage{stackengine}

\allowdisplaybreaks[4]
\usepackage[colorlinks = true,
linkcolor = blue,
urlcolor  = blue,
citecolor = blue,
anchorcolor = blue]{hyperref}

\usepackage[framemethod=TikZ]{mdframed}
\mdfdefinestyle{MyFrame}{%
	linecolor=black,
	outerlinewidth=0.5pt,
	roundcorner=1pt,
	innerrightmargin=5pt,
	innerleftmargin=5pt,}
	
\usepackage{graphicx}

\usepackage{diagbox}
\usepackage{colortbl}

\def\BibTeX{{\rm B\kern-.05em{\sc i\kern-.025em b}\kern-.08em
    T\kern-.1667em\lower.7ex\hbox{E}\kern-.125emX}}
\usepackage{cite}

\title{Quick Updates for the Perturbed Static Output Feedback Control Problem in Linear Systems with Applications to Power Systems}

\author{MirSaleh Bahavarnia$^\dagger$ and Ahmad F. Taha$^{\dagger,\star}$
	\thanks{$^\dagger$The authors are with the Department of Civil and Environmental Engineering, Vanderbilt University, TN 37235. $^\star$Dr. Taha has a secondary appointment with the Department of Electrical and Computer Engineering. Emails: \{mirsaleh.bahavarnia,ahmad.taha\}@vanderbilt.edu.}
 \thanks{This work is supported by the National Science Foundation (NSF) under Grants ECCS 2151571 and CMMI 2152450.}
}

\begin{document}

\maketitle

\begin{abstract}
This paper introduces a method for efficiently updating a nominal stabilizing static output feedback (SOF) controller in perturbed linear systems. As operating points and state-space matrices change in dynamic systems, accommodating updates to the SOF controller are necessary. Traditional methods address such changes by re-solving for the updated SOF gain, which is often (\textit{i}) computationally expensive due to the NP-hard nature of the problem or (\textit{ii}) infeasible due to the limitations of its semi-definite programming relaxations. To overcome this, we leverage the concept of \textit{minimum destabilizing real perturbation} (MDRP) to formulate a norm minimization problem that yields fast, reliable controller updates. This approach accommodates a variety of known perturbations, including abrupt changes, model inaccuracies, and equilibrium-dependent linearizations. We remark that the application of our proposed approach is limited to the class of SOF controllers in perturbed linear systems. We also introduce geometric metrics to quantify the proximity to instability and rigorously define stability-guaranteed regions. Extensive numerical simulations validate the efficiency and robustness of the proposed method. Moreover, such extensive numerical simulations corroborate that although we utilize a heuristic optimization method to compute the MDRP, it performs quite well in practice compared to an existing approximation method in the literature, namely the hybrid expansion-contraction (HEC) method. We demonstrate the results on the SOF control of multi-machine power networks with changing operating points, and demonstrate that the computed quick updates produce comparable solutions to the traditional SOF ones, while requiring orders of magnitude less computational time. 
\end{abstract}


\section{Introduction and Paper Contributions}\label{sec:Intro}

 \IEEEPARstart{T}{he static} output feedback (SOF) control problem, defined as follows for linear systems:
 $$\text{find stabilizing} \;\;\; F \;\;\; \text{for} \;\;\;  \dot{x} = Ax + B u, \;u = F y, \; y = C x$$
 is one of the simplest and most classic control engineering problems.  In particular, the SOF controller seeks to find $F=F^{\text{nominal}}$ that stabilizes the closed-loop system matrix $A+BFC$ given an unstable open-loop matrix $A$. This problem, though seemingly innocuous with a control architecture that is arguably the simplest among control routines, is known to be NP-hard due to the non-convexity in the feasible space of the SOF gain $F$. Over the past three decades, various approaches have been proposed to solve the SOF problem. Many rely on recasting the SOF problem as a convex semi-definite program (SDP) with linear matrix inequalities (LMI) constraints. 
 
 In this paper, given an initially stabilizing SOF $F=F^{\text{nominal}}$  controller, we focus on computing quick updates for the SOF problem $F^{\text{updated}}$ when the state-space matrices $A$, $B$, and $C$ change. The change can be due to either parametric uncertainty, improved updates of the system parameters, or a known change in the equilibrium point of a nonlinear system (thereby resulting in an updated small-signal, linearized model). The common approach to address this problem, when the operating point of the system is updated, is to simply re-solve the SOF and hence obtain a new $F^{\text{updated}}$. Here in this paper, we showcase that closed-form expressions for the solutions of the perturbed/updated SOF problem can be computed without the need to perform any intensive computations or optimization. This becomes very critical in application areas where the sampling time is small and updating the SOF gain for mid- to large-scale systems requires a long computational time. 
 
This problem is \textit{not} only a theoretical one. It is indeed relevant in power grids where the operating point changes every few minutes, as a result of updating optimal power flow (OPF) solutions, which yield perturbed or updated state-space matrices. A plethora of feedback control architectures in power systems relies on updating feedback gains once a new operating point is computed. More motivations and discussions for this application are given later in the paper. Throughout the paper, we assume that the perturbations or updated system matrices are known (as most linearizations yield known system matrices). Problems on robust control and unknown uncertainty are outside the scope of this work and hence are not discussed. In what follows is a brief literature review on this theoretical problem.  
 
\noindent \textbf{Brief Literature Review.} A key perspective to guide solving the perturbed SOF problem is through the notion of \textit{distance to instability}. This concept is a significant classical notion in control theory \cite{barnett1966insensitivity,chang1972adaptive,wong1977closed,patel1977robustness,patel1980quantitative,barrett1980conservatism,yedavalli1985improved,yedavalli1986reduced}. Distance to instability means how sensitive the stability of the control system is against the perturbations/uncertainties or, more generally, changes or updates in the state-space matrices. The varying nature of engineering systems' models necessitates the thorough analysis of distance to instability and its potential applications to develop robustly stable engineering systems. Several studies have quantitatively investigated the impacts of perturbations on the distance to instability of the control systems. In \cite{barnett1966insensitivity,wong1977closed}, a class of non-destabilizing linear constant perturbations is characterized for the linear-quadratic state feedback (LQSF) designs. The authors in \cite{chang1972adaptive} propose a guaranteed cost LQSF for which the closed-loop system is stable for any variation of a vector-valued parameter. In \cite{patel1977robustness}, for the LQSF designs, the distance-to-instability bounds are derived based on the algebraic Riccati equation (ARE) and Lyapunov stability theory. In \cite{patel1980quantitative}, bounds on the non-destabilizing time-varying nonlinear perturbations are obtained for asymptotically stable linear systems to provide computationally efficient quantitative measures. Various distance-to-instability tests are investigated in \cite{barrett1980conservatism} to highlight the trade-off between the distance-to-instability conservatism and the information about the perturbation. In \cite{yedavalli1985improved}, utilizing the Lyapunov stability theory, the author has proposed an improved non-destabilizing perturbation bound over the bound proposed by \cite{patel1980quantitative}. Taking advantage of appropriately chosen coordinate transformations, the authors in \cite{yedavalli1986reduced} have reduced the conservatism of non-destabilizing perturbation bounds proposed by \cite{patel1980quantitative,yedavalli1985improved}.

In contrast to the studies mentioned above, we do not go through the derivation of non-destabilizing perturbation bounds. Instead, we mainly focus on attenuating the impacts of perturbations on the system stability via \textit{updating} a nominal stabilizing static output feedback (SOF) controller. With that in mind, the control problem considered in this paper is an SOF controller \textit{update} problem.

As mentioned earlier, the SOF controller stabilization problem is known to be an NP-hard problem as it is intrinsically equivalent to solving a bi-linear matrix inequality (BMI) \cite{toker1995np}. Then, utilizing a typical approach by repeating the whole controller design procedure can become computationally cumbersome. Also, we do not utilize any Lyapunov-based approach (like approaches utilized by feedback stabilization research works \cite{petersen1987stabilization,khargonekar1990robust,xie1992h,peaucelle2005ellipsoidal,arastoo2016closed,jennawasin2021iterative,viana2023convex}) as it enforces an extra computational burden---mostly in the case of BMI and LMI formulations in SDPs \cite{majumdar2020recent}, which is not desired in terms of computational efficiency. Remarkably, the Lyapunov-based SOF controller synthesis hinges on approximately solving BMIs \cite{henrion2005solving,dinh2011combining} or incorporating sufficient LMI conditions \cite{crusius1999sufficient,arastoo2016closed}, which induces a conservatism. The alternative non-Lyapunov approach that we take is built upon the notion of \textit{minimum destabilizing real perturbation (MDRP)} \cite{van1984near}, which has inspired \cite{bahavarnia2017state,bahavarnia2019state} to synthesize sparse feedback controllers for the large-scale systems. The paper utilizes the fundamental linear algebraic results from \cite{horn2012matrix} where necessary.

\vspace{0.2cm}

\noindent {\bf Paper Contributions.} The main contributions of this paper can be itemized as follows:
\begin{itemize}
\item Built upon the notion of MDRP, we construct a simple norm minimization problem to propose a novel update of a nominal stabilizing SOF controller that can be applied to various control engineering applications in the case of perturbed scenarios like abrupt changes, inaccurate system models, or equilibrium-dependent linearized dynamics. We remark that the application of our proposed approach is limited to the class of SOF controllers in perturbed linear systems.
\item We propose novel updates of nominal stabilizing SOF controllers and derive sufficient stability conditions considering a known norm-bounded perturbation.
\item We define geometric metrics to quantitatively measure the distance-to-instability of the proposed updates of nominal stabilizing SOF controllers and characterize the corresponding guaranteed stability regions.
\item Through extensive numerical simulations, we validate the effectiveness of the theoretical results and present a thorough analysis of the numerical visualizations. We also demonstrate applications in feedback control of multi-machine power networks. We observe that utilizing the quick updates results in orders of magnitude faster updates for the SOF gain for even small-scale systems. For mid- to large-scale power systems, we note that the SDP-based approach to update the SOF gain requires tens of hours in computational time and is hence impractical. 
\end{itemize}

\noindent {\bf Paper Structure.} The remainder of the paper is structured as follows: Section \ref{sec:ProFor} motivates the main objective of the paper by raising a research question to be answered throughout the following sections. Section \ref{sec:Main} presents a novel updated stabilizing SOF controller via updating a nominal stabilizing SOF controller built upon a simple norm minimization problem (a least-squares problem). Section \ref{sec:MR} presents the main results of the paper while detailing the stability regions for the corresponding updated stabilizing SOF controllers. Through various numerical simulations, Section \ref{sec:Example} verifies the effectiveness of the theoretical results. Finally, the paper is concluded by drawing a few concluding remarks in Section \ref{Con}.

\noindent {\bf Paper Notation.} We denote the vectors and matrices by lowercase and uppercase letters, respectively. To represent the set of complex numbers, real numbers, positive real numbers, $n$-dimensional real-valued vectors, and $m \times n$-dimensional real-valued matrices, we respectively use $\mathbb{C}$, $\mathbb{R}$, $\mathbb{R}_{++}$, $\mathbb{R}^n$, and $\mathbb{R}^{m \times n}$. We represent the real part of a complex-valued number $z$ by $\Re(z)$. We denote the identity matrix of dimension $n$ with $I_n$. For a square matrix $M$, $\alpha(M)$ and $\Lambda(M)$ represent the spectral abscissa (i.e., the maximum real part of the eigenvalues) of $M$ and the spectrum of $M$, respectively. We say a square matrix $M$ is stable (Hurwitz) if $\alpha(M) < 0$ holds. For a matrix $M$, symbols $M^\top$, $\|M\|_F$, $\mathbf{vec}(M)$, and $U_M \Sigma_M V_M^\top$ denote its transpose, Frobenius norm, vectorization, and singular value decomposition (SVD), respectively. For a Laplace transfer function $T(s)$, we denote its $\mathcal{H}_{\infty}$ norm by $\|T(s)\|_{\mathcal{H}_{\infty}}$. Given a full-column rank matrix $M$, $M^{+} := (M^\top M)^{-1}M^\top$ denotes the Moore-Penrose inverse of $M$. We represent the Kronecker product with the symbol $\otimes$. For a vector $v$, we respectively denote its Euclidean norm and vectorization inverse with $\|v\|$ and $\mathbf{vec}^{-1}(v)$ where $\mathbf{vec}^{-1}(v)$ is a matrix that satisfies $\mathbf{vec}(\mathbf{vec}^{-1}(v)) = v$. We represent the set union with $\cup$. Given two real numbers $a < b$, we denote the open, closed, and half-open intervals with $]a,b[$, $[a,b]$, $[a,b[$, and $]a,b]$, respectively. We denote the computational complexity with big O notation, i.e., $\mathcal{O}()$. 

\section{Problem Statement} \label{sec:ProFor}

We consider the following linear state-space model:
\begin{align}
\dot{x}(t) &= (A + BF^{\text{nominal}}C) x(t) \label{SD}
\end{align}
where $x(t) \in \mathbb{R}^{n}$, $A \in \mathbb{R}^{n \times n}$, $B \in \mathbb{R}^{n \times m}$, $C \in \mathbb{R}^{p \times n}$, and $F^{\text{nominal}} \in \mathbb{R}^{m \times p}$ denote the state vector, state matrix, input matrix, output matrix, and a nominal stabilizing SOF controller matrix via the control law $u(t) = F^{\text{nominal}} y(t)$. The existence of $F^{\text{nominal}}$ entails that the closed-loop system is stable, i.e., $\alpha(A+BF^{\text{nominal}}C) < 0$ holds. 

Suppose that a known norm-bounded perturbation $\Delta \in \mathbb{R}^{n \times n}$ with an upper bound $\rho > 0$ on its Frobenius norm, (i.e., $0 < \|\Delta\|_F \le \rho$) hits the state-space model \eqref{SD} as follows:
\begin{align}
\dot{x}(t) &= (A + BF^{\text{nominal}}C + \Delta) x(t). \label{PerturbedSD}
\end{align} 
We note here an important observation. This perturbation $\Delta$ can model a variety of control engineering scenarios: (\textit{i}) changes in the operating/equilibrium point of a nonlinear system resulting in an updated state-space matrices $A$, $B$, and $C$, meaning that the $\Delta$ matrix captures the difference between the old and updated matrices; (\textit{ii}) some updates to the linear system state-space matrices when a parameter is changing in the system, necessitating knowing what this parametric change entails for the to-be-updated SOF gain; (\textit{iii}) topological changes in the dynamic network structure when connections between nodes are either established or removed. The three scenarios (\textit{i})--(\textit{iii}) have an abundance of applications in control engineering, and all boil down to requiring updating the feedback gain $F^{\text{nominal}}$ to accommodate changes in the state-space matrices and hence ensuring closed-loop system stability. In all three scenarios, matrix $\Delta$ is assumed to be known---a reasonable assumption in scenarios (\textit{i})--(\textit{iii}). 

Similar to \cite{patel1977robustness,patel1980quantitative,barrett1980conservatism,van1984near,yedavalli1986reduced}, we choose the Frobenius norm over the spectral norm as it provides more analytic convenience. On the one hand, for non-destabilizing perturbations (e.g., sufficiently small perturbations), although $A + BF^{\text{nominal}}C + \Delta$ in \eqref{PerturbedSD} is still a stable matrix, the distance to instability can be degraded, which is undesired. On the other hand, for destabilizing perturbations (e.g., more severe perturbations), $A + BF^{\text{nominal}}C + \Delta$ in \eqref{PerturbedSD} can become unstable, which is even more undesired. To attenuate the impacts of such perturbations on the distance-to-instability and the stability, a \textit{typical} approach can be repeating the whole controller design procedure to find a new SOF controller, namely $F^{\text{typical}}$, to stabilize $A + \Delta$ and get a stable $A + \Delta + B F^{\text{typical}} C$. Such a typical approach can be computationally inefficient and even infeasible in some cases. 

In general, the SOF controller stabilization problem 
is known to be an NP-hard problem as it is intrinsically equivalent to solving a BMI \cite{toker1995np}
\begin{align*}
    & (A+B\textcolor{red}{F}C)^\top \textcolor{red}{P} + \textcolor{red}{P}(A+B\textcolor{red}{F}C) \prec 0
\end{align*}
for $\textcolor{red}{P} \succ 0$ and $\textcolor{red}{F}$ which is computationally challenging for large-scale systems. Moreover, the alternative sufficient condition-based LMI \cite{crusius1999sufficient}
\begin{align*}
    A^\top \textcolor{blue}{P} + \textcolor{blue}{P} A + C^\top \textcolor{blue}{N}^\top B^\top + B \textcolor{blue}{N} C \prec 0,\;\;\;
    B \textcolor{blue}{M} = \textcolor{blue}{P} B
\end{align*}
may likely either end up with an infeasible situation for $(\textcolor{blue}{P},\textcolor{blue}{N},\textcolor{blue}{M})$ and $\textcolor{blue}{F} = \textcolor{blue}{M}^{-1} \textcolor{blue}{N}$ or computational burden in the case of large-scale systems. Indeed, solving the LMIs for a system with only a few hundred states on standard off-the-shelf SDP solvers like SDPT3 \cite{toh1999sdpt3} and MOSEK \cite{andersen2000mosek} takes many hours, which is often far longer than the sampling time of the system. 

Motivated by issues mentioned above and utilizing a simple norm minimization problem built upon the notion of MDRP \cite{van1984near}, we propose a novel update of a nominal stabilizing SOF controller that can be applied to various control engineering applications in the case of perturbed scenarios like abrupt changes, inaccurate system models, or equilibrium-dependent linearized dynamics. We remark that the application of our proposed approach is limited to the class of SOF controllers in perturbed linear systems and emphasize that the stability analysis from the general perspective of the switched nonlinear systems is beyond the scope of this paper. We meanwhile admit the fact that for the case of general unknown perturbations, various well-developed robust control tools have already been presented in the literature. In a nutshell, the main objective of this paper is to find an answer to the following research question:

\vspace{0.2cm}

\noindent \textit{\textbf{Q1:} Given the perturbed state-space model \eqref{PerturbedSD}, how can we update a nominal stabilizing SOF controller $F^{\text{nominal}}$ such that the closed-loop system remains stable?}

\section{Updating the Nominal Stabilizing SOF Controller} \label{sec:Main}

This section consists of twofold: (\textit{i}) motivation and \textit{(ii) main idea}. First, we present what motivates us to propose a novel update of a nominal stabilizing SOF controller. Second, we detail the main idea behind the proposed updated stabilizing SOF controller.
\subsection{Motivation}

To improve the distance to instability of the perturbed state-space \eqref{PerturbedSD}, let us consider the updated stabilizing SOF controller, as $F^{\text{nominal}} + G$ ($G$: update matrix), with the following state-space model:
\begin{align}
\dot{x}(t) &= (A + \Delta + B(F^{\text{nominal}} + G)C) x(t). \label{RDPerturbedSD}
\end{align}
For instance, for the special case of the typical approach, $G^{\text{typical}} = F^{\text{typical}} - F^{\text{nominal}}$ holds and for the case of our proposed method, we have $G^{\text{updated}} = F^{\text{updated}} - F^{\text{nominal}}$. For simplicity, we drop the superscript from $F^{\text{nominal}}$ and simply denote it by $F$ from now on.

Defining the notion of \textit{minimum destabilizing real perturbation} (MDRP) of a given stable matrix $\mathcal{A} \in \mathbb{R}^{n \times n}$, namely $\beta_{\mathbb{R}}(\mathcal{A})$, as follows ($(3.2)$ in \cite{van1984near}):
\begin{align*}
\beta_{\mathbb{R}}(\mathcal{A}) &:= \min \{ \| \mathcal{X} \|_F: \alpha(\mathcal{A} + \mathcal{X}) = 0, \mathcal{X} \in \mathbb{R}^{n \times n} \}
\end{align*}
and choosing $\mathcal{A} = A + BFC$ and $\mathcal{X} = BGC + \Delta$ based on the updated perturbed state-space model \eqref{RDPerturbedSD}, we see that if
\begin{align} \label{ImpIneq}
\|BGC + \Delta\|_F &< \beta_{\mathbb{R}}(A+BFC)
\end{align} 
holds, then $A + \Delta + B(F + G)C$ is stable, i.e., $F +G$ is an updated stabilizing SOF controller for $A + \Delta$. Inequality \eqref{ImpIneq} motivates us to search for an efficient update $F + G$ via minimizing the $\|BGC + \Delta\|_F$.

In the sequel, we present the lower and upper bounds on MDRP of $A + BFC$, followed by a brief description of its exact value computation. It is noteworthy that the MDRP is also known as the real stability radius (RSR) \cite{hinrichsen1986stability}.
 
\subsubsection{Lower bound}Considering the fact that $\alpha(X)$ is a continuous function with respect to $X$, we have by definition
\begin{align*}
& \forall \epsilon > 0, \exists \delta(\epsilon) > 0,~\mathrm{s.t.}~\mathrm{if}~\|\mathcal{X}\|_F < \delta(\epsilon)~\mathrm{holds},~\mathrm{then} \notag \\
& \alpha(\mathcal{A}) - \epsilon < \alpha(\mathcal{A} + \mathcal{X}) < \alpha(\mathcal{A}) + \epsilon~\mathrm{holds}.
\end{align*}
Then, choosing $\mathcal{A} = A + BFC$ and $\mathcal{X} = BGC + \Delta$, we realize that for any $\epsilon$ satisfying $\epsilon < -\alpha(A+BFC)$, if $\|BGC + \Delta\|_F < \delta(\epsilon)$ holds, then $A + \Delta + B(F + G)C$ is stable. That suggests the following lower bound on MDRP of $A + BFC$:\begin{subequations} \label{LoBo}
\begin{align}
& 0 < \delta_{\sup} \le \beta_{\mathbb{R}}(A+BFC)\\
& \delta_{\sup} := \sup \{\delta(\epsilon): \epsilon \in ]0,-\alpha(A+BFC)[\}. 
\end{align}
\end{subequations}

Since the RCR is lower bounded by the complex stability radius (CSR) and it is a well-known fact that $\mathrm{CSR} = \frac{1}{\|(sI_n-(A+BFC))^{-1}\|_{\mathcal{H}_{\infty}}}$ holds \cite{hinrichsen1990real}, we then get the following lower bound on MDRP of $A+BFC$:
\begin{align} \label{CSR}
    \underbrace{\frac{1}{\|(sI_n-(A+BFC))^{-1}\|_{\mathcal{H}_{\infty}}}}_{\beta_{\mathbb{R}}^l} \le \beta_{\mathbb{R}}(A+BFC).
\end{align}

\subsubsection{Upper bound}On one hand, since $\alpha(\mathcal{A} + \mathcal{X}) = 0$ holds for the choice of $\mathcal{X} = -\alpha(\mathcal{A}) I_n$, then choosing $\mathcal{A} = A + BFC$ and $\mathcal{X} = -\alpha(A+BFC) I_n$, we get the following upper bound on MDRP of $A + BFC$ \cite{van1984near}:
\begin{align} \label{UpBo}
& \beta_{\mathbb{R}}(A+BFC) \le -\sqrt{n}\alpha(A+BFC).
\end{align}
On the other hand, given $\mathcal{A} = U_{\mathcal{A}} \Sigma_{\mathcal{A}} V_{\mathcal{A}}^\top$ as the singular value decomposition (SVD) of $\mathcal{A}$ and choosing $\mathcal{X} = - \sigma_{\mathcal{A}}^{\min} u_\mathcal{\mathcal{A}}^{\min} v_{\mathcal{A}}^{\min T}$ (superscript $\min$ denotes the corresponding minimum singular value and vectors), it can be verified that $\alpha(\mathcal{A} + \mathcal{X}) = 0$ holds. Then, choosing $\mathcal{A} = A + BFC$ and according to \eqref{UpBo}, we get the following upper bound on MDRP of $A + BFC$ \cite{van1984near}:\begin{subequations} \label{UpBoG}
\begin{align}
& \beta_{\mathbb{R}}(A+BFC) \le \beta_{\mathbb{R}}^{u}\\
& \beta_{\mathbb{R}}^{u} = \min \{\sigma^{\min}(A+BFC),-\sqrt{n}\alpha(A+BFC)\}.
\end{align}
\end{subequations}
For the special case of a symmetric matrix $A+BFC$, since $\sigma^{\min}(A+BFC) = -\alpha(A+BFC)$ holds, \eqref{UpBoG} reduces to
\begin{align} \label{UpBoDiag}
& \beta_{\mathbb{R}}(A+BFC) \le -\alpha(A+BFC)
\end{align}
which is a tighter bound compared to the upper bound in \eqref{UpBo}. According to Corollary $3.5.$ in \cite{hinrichsen1986stability} and noting that $\|\mathcal{X}\| \le \|\mathcal{X}\|_F$ holds for any $\mathcal{X}$ \cite{horn2012matrix}, it can be verified that $\sigma^{\min}(A+BFC) \le \beta_{\mathbb{R}}(A+BFC)$ holds and the equality in \eqref{UpBoDiag} is consequently satisfied.
 
\subsubsection{Exact value} Unfortunately, computing the exact value of $\beta_{\mathbb{R}}(A+BFC)$ is computationally intractable (NP-hard), and no practical estimation technique could be devised due to the difficulty of the ensuing constrained minimization problem cast by \cite{van1984near}. Also, there is no systematic tractable way to compute the exact value of the lower bound $\delta_{\sup}$ in \eqref{LoBo} since we only know about the existence of $\delta(\epsilon)$ and nothing more. However, taking advantage of the lower and upper bounds on $\beta_{\mathbb{R}}(A+BFC)$ (derived in \eqref{CSR} and \eqref{UpBoG}, respectively), we may utilize heuristics to obtain an appropriate approximate value of $\beta_{\mathbb{R}}(A+BFC)$ in a reasonable computational time. Since \eqref{ImpIneq} plays a significant role in the characterization of the stability regions, the tightness of the upper bound on $\beta_{\mathbb{R}} (A+BFC)$ in \eqref{UpBoG} becomes crucial. Remarkably, if the equality in \eqref{UpBoG} becomes active (i.e., the case of a tight upper bound), then the proposed updated stabilizing SOF controller in this paper becomes efficient as it only requires the value of $\beta_{\mathbb{R}}^{u}$ which can efficiently be computed (e.g., the case of a symmetric $A+BFC$ for which $\beta_{\mathbb{R}}(A+BFC) = -\alpha(A+BFC)$ holds). For the special case of structured perturbation, i.e., $\Delta = B M C$ for a matrix $M \in \mathbb{R}^{m \times p}$, one may compute MDRP via (frequency domain)-based algorithms detailed by \cite{hinrichsen1990real}. 

Guglielmi et al. \cite{guglielmi2017approximating} have proposed a hybrid expansion-contraction (HEC)-based method, namely Algorithm HEC-RF: HEC real Frobenius norm, to approximate the RSR. Precisely, the expansion phase is solved via the algorithm SVSA-RF (spectral value set abscissa: real Frobenius norm), and the contraction phase is solved via a Newton-bisection zero-finding algorithm. The HEC-based proposed method is quadratically convergent and efficient for large-scale and sparse linear dynamical systems. Moreover, their extensive numerical experiments reveal that the approximated RSR coincides with the exact RSR in most cases. Due to such a promising performance of the HEC-based method, one can utilize it as an efficient method to compute the RSR (i.e., MDRP).
 
\subsection{Main idea of the updated gain}

Since \eqref{ImpIneq} provides a sufficient condition on the stability of $A + \Delta + B(F + G)C$, our main idea to propose an efficient updated stabilizing SOF controller $F + G$ is to compute $G$ via minimizing $\|BGC + \Delta\|_F^2$ and verifying that under which conditions, the minimized value of $\|BGC + \Delta\|_F^2$ would be less than $\beta_{\mathbb{R}}(A+BFC)^2$. It is noteworthy that if the most optimistic scenario occurs, (i.e., the scenario in which for a known $\Delta$, equation $\|BGC + \Delta\|_F = 0$ has a solution $G$), then one can completely cancel out the effect of the hitting perturbation $\Delta$ and retrieve the primary unperturbed $A + BFC$ as detailed later on. With that in mind and to find a reasonable answer to the question stated in Section \ref{sec:ProFor} (\textit{Q1}), we consider the following optimization problem:
\begin{align} \label{OP}
& \underset{G \in \mathbb{R}^{m \times p}}{\min}~\|BGC + \Delta\|_F^2.
\end{align}
By vectorizing $BGC + \Delta$, defining $g := \mathbf{vec}(G), \delta := \mathbf{vec}(\Delta), H := C^\top \otimes B$, and noting that $\mathbf{vec}(\mathcal{XYZ}) = (\mathcal{Z}^\top \otimes \mathcal{X}) \mathbf{vec}(\mathcal{Y})$ holds for any triplet $(\mathcal{X},\mathcal{Y},\mathcal{Z})$ with consistent dimensions and $\| \mathbf{vec}(X)\| = \| X \|_F$ holds for any $X$, optimization problem \eqref{OP} can equivalently be cast as the following least-squares problem \cite{boyd2004convex}:
\begin{align} \label{EOP}
& \underset{g \in \mathbb{R}^{mp}}{\min}~\|Hg + \delta\|^2.
\end{align}
In this paper, we assume that the following standard assumption holds for $B$ and $C$.
\begin{asmp} \label{Asmp1}
    We assume that $B$ and $C$ are full-column rank and full-row rank, respectively.
\end{asmp}
\noindent According to Assumption \ref{Asmp1} and noting that identity $(C^\top \otimes B)^{+} = {C}^{\top+} \otimes B^{+}$ holds, optimization problem \eqref{EOP} can analytically be solved as
\begin{align} \label{goptsol}
g^{\ast}_{\delta} &= -({C}^{\top +} \otimes B^{+})\delta
\end{align} 
and the analytic optimal solution of \eqref{OP} can subsequently be presented as follows:
\begin{align} \label{Osol}
G^{\ast}_{\Delta} &= \mathbf{vec}^{-1}(g^{\ast}_{\delta}) = -B^{+} \Delta ({C}^{\top +})^\top
\end{align}
for which the computational complexity is $\mathcal{O}(n^2 \min \{m,p\})$ while the computational complexity of \eqref{goptsol} is $\mathcal{O}(n^2 m^2 p^2)$. Substituting $g^{\ast}_{\delta}$ of \eqref{goptsol} in \eqref{EOP}, the optimal value of the objective function in \eqref{EOP}, namely $J^{\ast}(\delta)$, becomes
\begin{align} \label{OpVal}
& J^{\ast}(\delta) := \|H g^{\ast}_{\delta} + \delta\|^2 =  \| (I_{n^2} - HH^{+}) \delta\|^2.
\end{align}
Defining $P := I_{n^2} - H H^{+}$ and noting that $P^\top P = P$ holds (since $H^{+}H = I_{mp}$ holds), \eqref{OpVal} reduces to
\begin{align} \label{ROV}
J^{\ast}(\delta) &= \delta^\top P \delta.
\end{align}
Precisely, with a bit of abuse of notation, we define $J^{\ast}(\Delta) := J^{\ast}(\mathbf{vec}(\Delta)) = J^{\ast}(\delta)$.

\section{Main Results} \label{sec:MR}

This section includes the main results of the paper. Given a known norm-bounded perturbation $\Delta$ with $0 < \| \Delta \|_F \le \rho$, we investigate the dependency of $J^{\ast}(\Delta)$ on $\Delta$ via inspecting the linear algebraic properties of $P$ in \eqref{ROV}. Proposition \ref{Prop1} analytically parameterizes the norm-bounded perturbation and proposes a closed-form formula for $J^{\ast}(\Delta)$. Proposition \ref{Prop2} elaborates on deriving sufficient conditions for the stability of the proposed updated stabilizing SOF controllers while analytically characterizing the guaranteed stability regions. Furthermore, we define a geometric metric to quantify the distance-to-instability of the proposed updated stabilizing SOF controllers. In the sequel, to save space, whenever needed, we refer to $\beta_{\mathbb{R}}(A+BFC)$ as $\beta$.

In the following lemma, we present an SVD-based parameterization of $P$ in \eqref{ROV} that facilitates parameterizing the norm-bounded perturbation $\Delta$ and subsequently proposing a closed-form expression for $J^{\ast}(\Delta)$.

\begin{mylem} \label{Lem1}
Suppose that $H = U_H \Sigma_H V_H^\top$ is the SVD of $H$. Then, $P$ in \eqref{ROV} can be parameterized as follows:
\begin{align} \label{Pform}
P &= U_H \begin{bmatrix} 0 & 0\\0 & I_{n^2 - mp} \end{bmatrix} U_H^\top
\end{align}
where $U_H = (V_C \otimes U_B)U_{\Omega}$ holds provided that $B = U_B \Sigma_B V_B^\top$, $C = U_C \Sigma_C V_C^\top$, and $\Omega := \Sigma_C^\top \otimes \Sigma_B = U_{\Omega} \Sigma_{\Omega} V_{\Omega}^\top$ denote the SVDs of $B$, $C$, and $\Omega$, respectively.
\end{mylem}
\begin{proof}
See Appendix \ref{ApndxA}.
\end{proof}

\subsection{Norm-bounded perturbation analytic parameterization} Built upon Lemma \ref{Lem1}, we present the following proposition that analytically parameterizes the norm-bounded perturbation $\Delta$ while proposing a closed-form expression for $J^{\ast}(\Delta)$. 

\begin{myprs} \label{Prop1}
Given the norm-bounded perturbation $\Delta$ with $\|\Delta\|_F = r$ and $r \in ]0,\rho]$, and considering $r = \rho \sin (\frac{\pi \tau}{2})$ with $\tau \in ]0,1]$, the norm-bounded perturbation $\Delta$ can be parameterized as follows:
\begin{align} \label{MR}
\Delta &= \rho \sin \Big(\frac{\pi \tau}{2}\Big) U_B \mathbf{vec}^{-1}\bigg(U_{\Omega}\begin{bmatrix} \phi_c \cos(\frac{\pi \theta}{2}) \\ \phi_s \sin(\frac{\pi \theta}{2}) \end{bmatrix}\bigg) V_C^\top
\end{align}
where $\phi_c \in \mathbb{R}^{mp}$ with $\|\phi_c\| =1$, $\phi_s \in \mathbb{R}^{n^2-mp}$ with $\|\phi_s\| =1$, and $\theta \in [0,1]$, and we can compute $J^{\ast}(\Delta)$ in \eqref{ROV} as follows:
\begin{align} \label{JMR}
J^{\ast}(\Delta) &= \bigg(\rho \sin \Big(\frac{\pi \tau}{2}\Big) \sin \Big(\frac{\pi \theta}{2}\Big) \bigg)^2.
\end{align}
\end{myprs}

\begin{proof}
See Appendix \ref{ApndxB}.
\end{proof}
The following corollary provides an alternative formula to compute $G^{\ast}_{\Delta}$ in \eqref{Osol}.
\begin{mycor}
Considering the following identities:
\begin{align*} 
& (U_H,\Sigma_H,V_H) = ((V_C \otimes U_B)U_{\Omega},\Sigma_{\Omega},(U_C \otimes V_B)V_{\Omega})\\
& B = U_B \Sigma_B V_B^\top, C = U_C \Sigma_C V_C^\top, \Sigma_C^\top \otimes \Sigma_B = U_{\Omega} \Sigma_{\Omega} V_{\Omega}^\top
\end{align*}
\eqref{Osol} can alternatively be computed as follows:
\begin{align*}
& G^{\ast}_{\Delta} = -\mathbf{vec}^{-1}(V_H \begin{bmatrix} (\begin{bmatrix} I_{mp} & 0 \end{bmatrix}\Sigma_H)^{-1} & 0 \end{bmatrix} U_H^\top \mathbf{vec}(\Delta)).
\end{align*}
\end{mycor}

Fig. \ref{Fig0} depicts the dependency of $\frac{J^{\ast}(\Delta)}{\rho^2}$ on $\tau$ and $\theta$. As expected, since functions $\sin(\frac{\pi \tau}{2})$ and $\sin(\frac{\pi \theta}{2})$ have monotonic behaviors versus $\tau$ (for $\tau \in ]0,1]$) and $\theta$ (for $\theta \in [0,1]$), respectively, the smaller $\tau$ and/or $\theta$, the smaller $\frac{J^{\ast}(\Delta)}{\rho^2}$ we get. Note that the smaller value of $\frac{J^{\ast}(\Delta)}{\rho^2}$ is equivalent to the higher chance of satisfaction of the sufficient stability condition \eqref{ImpIneq}. In other words, its intuitive interpretation is that handling a less severe perturbation via an updated stabilizing SOF controller $F + G^{\ast}_{\Delta}$ with $G^{\ast}_{\Delta}$ in \eqref{Osol} is easier.

\begin{figure}[t]
\centering
\includegraphics[scale=0.4]{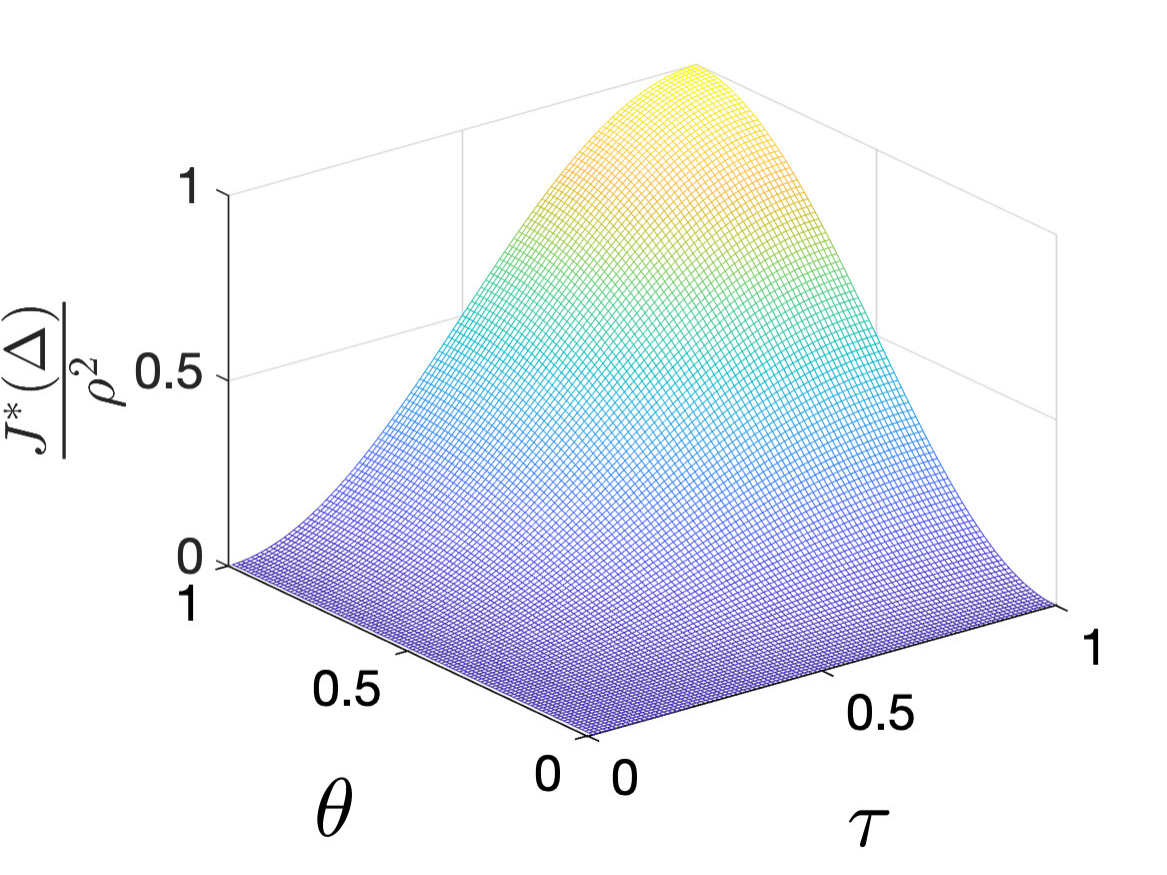}
\caption{Dependency of $\frac{J^{\ast}(\Delta)}{\rho^2}$ on $\tau$ and $\theta$.} \label{Fig0}
\end{figure}

\subsection{Guaranteed stability region analytic characterization} We state the following proposition that derives sufficient conditions on the stability of the proposed updated stabilizing SOF controllers while analytically characterizing the guaranteed stability regions.

\begin{myprs} \label{Prop2}
Given the norm-bounded perturbation $\Delta$ parameterized by \eqref{MR}, $F + G^{\ast}_{\Delta}$ with $G^{\ast}_{\Delta}$ in \eqref{Osol} is an updated stabilizing SOF controller
\begin{enumerate}
\item if $\rho < \beta_{\mathbb{R}}(A+BFC)$ holds; 
\item else if $\rho \ge \beta_{\mathbb{R}}(A+BFC)$ and $(\tau_{\Delta},\theta_{\Delta}) \in S_{\kappa}$ hold where the guaranteed stability region $S_{\kappa}$ is defined as:\begin{subequations} \label{SSS}
\begin{align}
S_{\kappa} &:= \check{S} \cup \tilde{S}\\
\check{S} &:= \{(\tau,\theta): \tau \in ]0,\kappa[, \theta \in [0,1]\}\\
\kappa &:= \frac{2}{\pi} \arcsin\Big(\frac{\beta_{\mathbb{R}}(A+BFC)}{\rho}\Big)\\
\tilde{S} &:= \{(\tau,\theta): \tau \in [\kappa,1], \theta \in [0,\zeta_{\tau,\kappa}[\}\\
\zeta_{\tau,\kappa} &:= \frac{2}{\pi}\arcsin\Big(\frac{\sin (\frac{\pi \kappa}{2})}{\sin (\frac{\pi \tau}{2})}\Big).
\end{align}
\end{subequations}
Moreover, the following geometric metric provides a percentage-based lower bound on the stability of the updated perturbed state-space \eqref{RDPerturbedSD}:
\begin{align} \label{AreaM}
\xi_{\kappa}~(\%) &:= 100 \times \bigg ( \kappa + \int_{\kappa}^{1} \zeta_{\tau,\kappa} d \tau \bigg )
\end{align}
and $\xi_{\kappa}$ is an increasing function of $\kappa$ (equivalently $\xi_{\rho}$ is a decreasing function of $\rho$ for a fixed $\beta_{\mathbb{R}}(A+BFC)$ and an increasing function of $\beta_{\mathbb{R}}(A+BFC)$ for a fixed $\rho$).
\end{enumerate}
\end{myprs}

\begin{proof}
See Appendix \ref{ApndxC}.
\end{proof}

For the case of $\rho < \beta_{\mathbb{R}}(A+BFC)$, the guaranteed stability region would be $]0,1] \times [0,1] = S_{\kappa}|_{\kappa = 1} \cup \{ (1,1) \}$, i.e., the unit square in the non-negative quadrant of $(\tau,\theta)$. For the sake of notation simplicity, we define $\mathbb{S} = ]0,1] \times [0,1]$ and utilize the unified notation of $S$ to refer to both guaranteed stability regions $S_{\kappa}$ and $\mathbb{S}$. The following corollary thoroughly sheds light on the dependency and limiting behaviors of $\xi_{\rho}$ and $\xi_{\beta}$ on $\rho$ and $\beta$, respectively.
\begin{mycor} \label{Coro2}
For the case of $\rho \ge \beta_{\mathbb{R}}(A+BFC)$, considering the following expression for $\xi_{\rho}$:
\begin{align*}
\xi_{\rho} &= \frac{2}{\pi} \arcsin\Big(\frac{\beta}{\rho}\Big) + \frac{2}{\pi} \int_{\frac{2}{\pi} \arcsin\big(\frac{\beta}{\rho}\big)}^{1} \arcsin\Big(\frac{\beta}{\rho \sin (\frac{\pi \tau}{2})}\Big) d\tau 
\end{align*}
we compute the derivative of $\xi_{\rho}$ with respect to $\rho$ as follows:
\begin{align} \label{Deri}
\frac{d \xi_{\rho}}{d\rho} &= -\frac{2}{\pi \rho} \int_{\frac{2}{\pi} \arcsin\big(\frac{\beta}{\rho}\big)}^{1} \frac{\beta}{\rho \sqrt{\sin (\frac{\pi \tau}{2})^2 - (\frac{\beta}{\rho})^2}} d \tau.
\end{align}
Moreover, as $\rho$ tends to $\beta$ and $\infty$ in \eqref{Deri}, we get
\begin{align*}
& \lim_{\rho \to \beta^+} \frac{d \xi_{\rho}}{d\rho} = -\frac{2}{\pi \beta}, \lim_{\rho \to \infty} \frac{d \xi_{\rho}}{d\rho} = 0, \lim_{\rho \to \beta^+} \xi_{\rho} = 1, \\ & \lim_{\rho \to \infty} \xi_{\rho} = 0.
\end{align*}
Similarly, considering the following expression for $\xi_{\beta}$:
\begin{align*}
\xi_{\beta} &= \frac{2}{\pi} \arcsin\Big(\frac{\beta}{\rho}\Big) + \frac{2}{\pi} \int_{\frac{2}{\pi} \arcsin\big(\frac{\beta}{\rho}\big)}^{1} \arcsin\Big(\frac{\beta}{\rho \sin (\frac{\pi \tau}{2})}\Big) d\tau
\end{align*}
we compute the derivative of $\xi_{\beta}$ with respect to $\beta$ as follows:
\begin{align} \label{Deri2}
\frac{d \xi_{\beta}}{d\beta} &= \frac{2}{\pi \rho} \int_{\frac{2}{\pi} \arcsin\big(\frac{\beta}{\rho}\big)}^{1} \frac{1}{\sqrt{\sin (\frac{\pi \tau}{2})^2 - (\frac{\beta}{\rho})^2}} d \tau.
\end{align}
Moreover, by tending $\beta$ to $0$ and $\rho$ in \eqref{Deri2}, we get
\begin{align*}
& \lim_{\beta \to 0^+} \frac{d \xi_{\beta}}{d\beta} = \infty, \lim_{\beta \to \rho^{-}} \frac{d \xi_{\beta}}{d\beta} = \frac{2}{\pi \rho}, \lim_{\beta \to 0^+} \xi_{\beta} = 0,\\ & \lim_{\beta \to \rho^{-}} \xi_{\beta} = 1.
\end{align*}
\end{mycor}

Fig. \ref{Fig1} visualizes the guaranteed stability region $S_{\kappa}$ for $\kappa = \frac{1}{3}$ and the percentage-based lower bounds on the stability of the updated perturbed state-space \eqref{RDPerturbedSD} versus $\kappa$, $\rho$, and $\beta$. As expected, the numerical observations of Fig. \ref{Fig1} are consistent with the theoretical results of Proposition \ref{Prop2} and Corollary \ref{Coro2}. Precisely, as $\kappa$ decreases, e.g., for an increased perturbation upper bound $\rho$ or a decreased MDRP $\beta$, the percentage-based lower bound on the stability of the updated perturbed state-space \eqref{RDPerturbedSD} $\xi~(\%)$ degrades, which is expected. As Fig. \ref{Fig1} (Top-Left) depicts, for the sufficiently large values of $\tau$ and/or $\theta$, i.e., more severe perturbations, $(\tau,\theta)$ lies outside the $S_{\kappa}$ and there is no stability guarantee for the proposed updated SOF controller which is aligned with the expectations around the negative impacts of perturbations on the stability. 

\begin{figure}[t]
\centering
\subfloat[a][]{\includegraphics[scale=0.225]{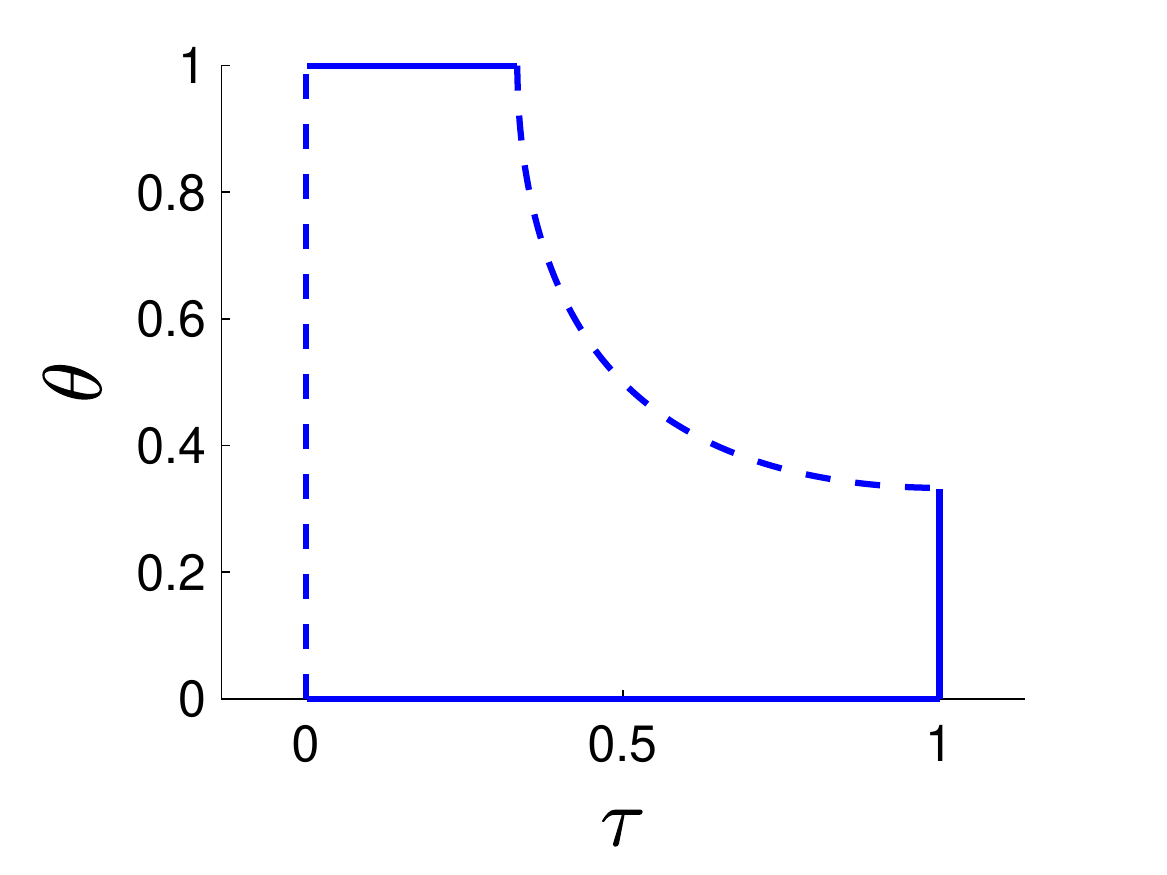}\label{<figure1>}}
\subfloat[b][]{\includegraphics[scale=0.225]{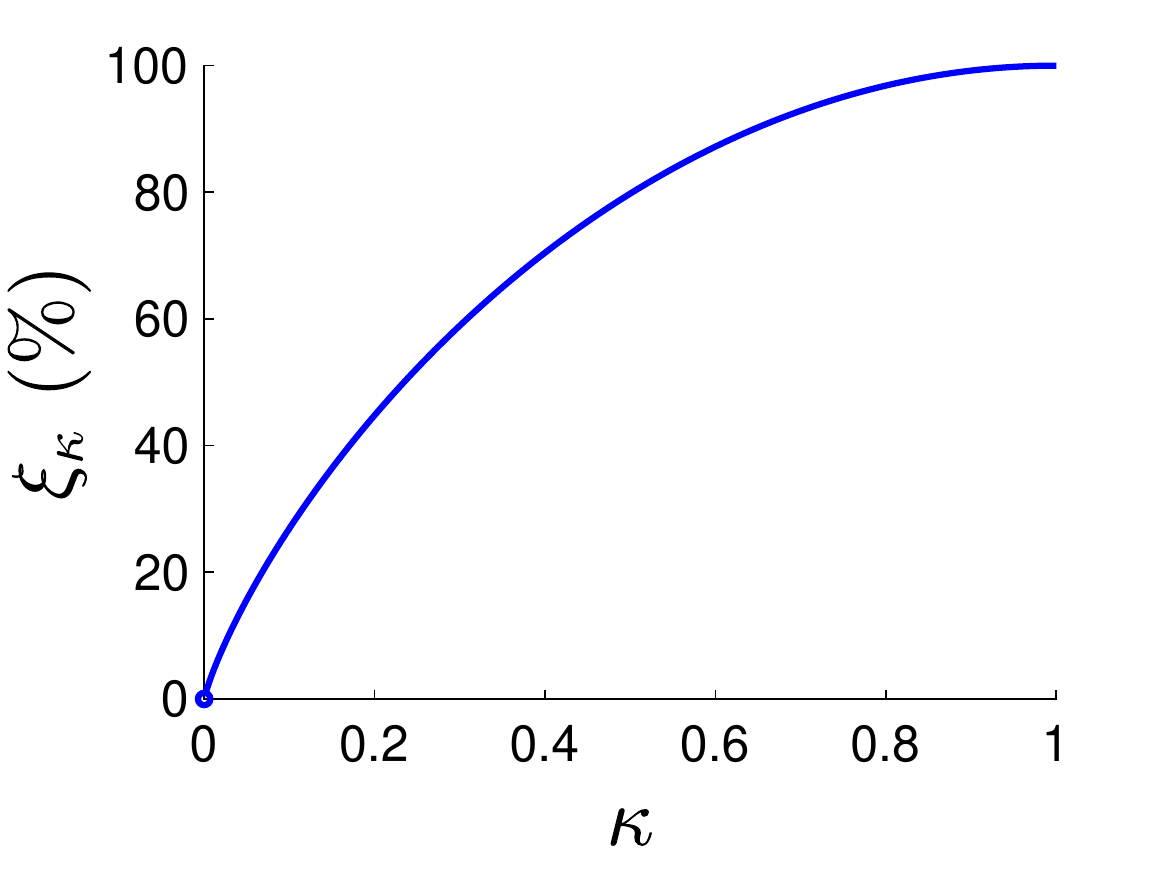}\label{<figure2>}}

\subfloat[c][]{\includegraphics[scale=0.225]{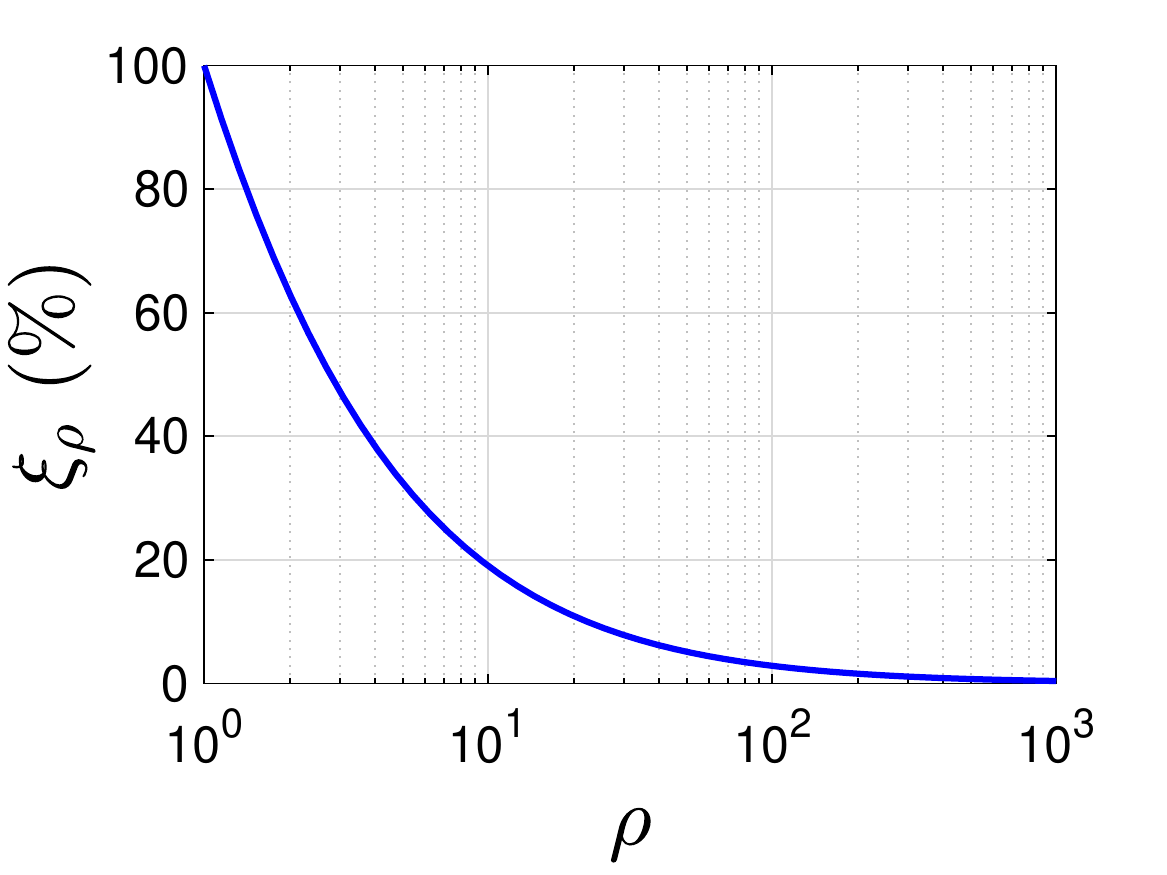}\label{<figure3>}}
\subfloat[d][]{\includegraphics[scale=0.225]{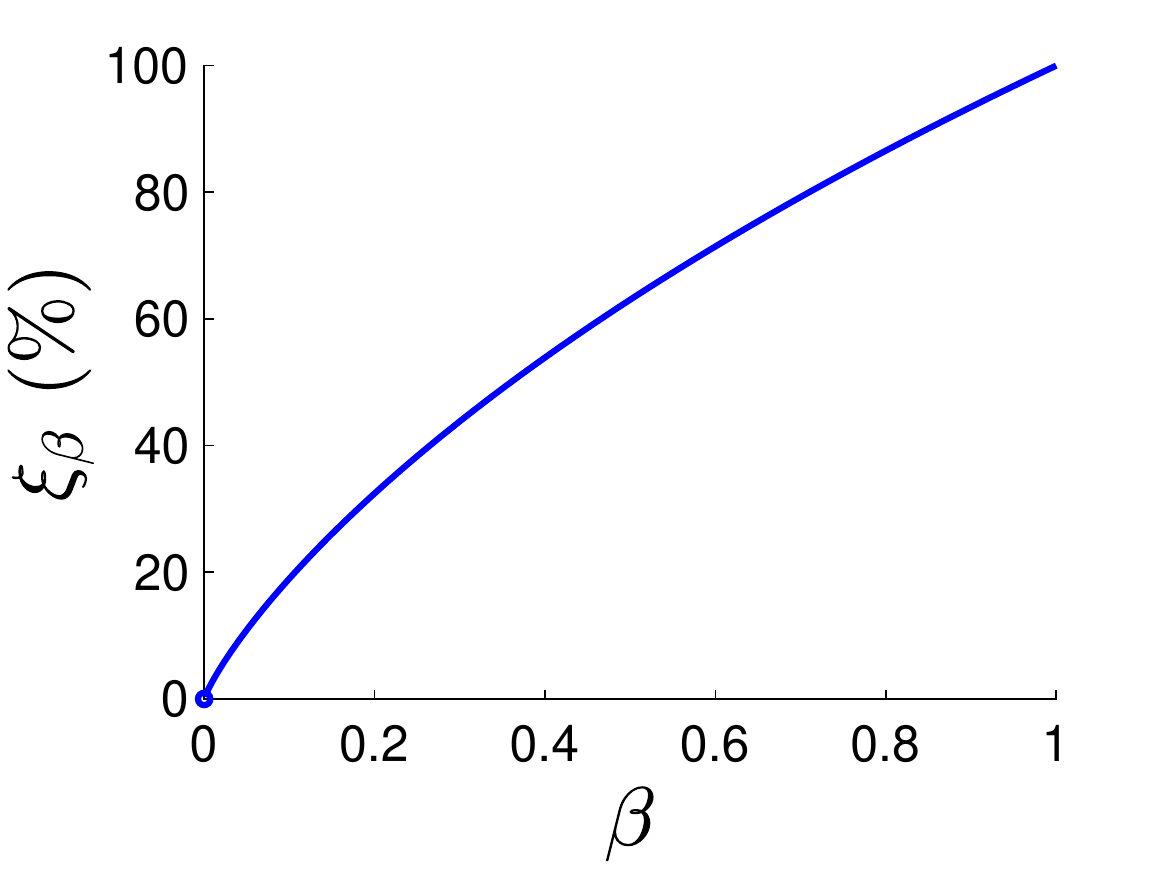}\label{<figure4>}}

\caption{(a) The guaranteed stability region $S_{\kappa}$ for $\kappa = \frac{1}{3}$, (b) the percentage-based lower bound on the stability of the updated perturbed state-space \eqref{RDPerturbedSD} $\xi_{\kappa}~(\%)$ versus $\kappa$, (c) the percentage-based lower bound on the stability of the updated perturbed state-space \eqref{RDPerturbedSD} $\xi_{\rho}~(\%)$ versus $\rho$ for $\beta = 1$, and (d) the percentage-based lower bound on the stability of the updated perturbed state-space \eqref{RDPerturbedSD} $\xi_{\beta}~(\%)$ versus $\beta$ for $\rho = 1$.} \label{Fig1}
\end{figure}

\section{Numerical Simulations} \label{sec:Example}

First, we elaborate on the computation of the MDRP as it constitutes the core part of Procedure \ref{proced}, and the performance of Procedure \ref{proced} is significantly affected by the computational accuracy/efficiency of the MDRP. Then, to assess the effectiveness of the theoretical results, we consider two scenarios: (\textit{i}) a toy example with an SOF controller benchmark collected by \cite{leibfritz2006compleib}, (\textit{ii}) two power system examples with standard benchmarks from \cite{bazr2017,bazrafshan2018coupling,11010910}. We will use Procedure \ref{proced} to compute the updated stabilizing SOF controller $F^{\text{updated}}$ from a nominal stabilizing SOF controller $F^{\text{nominal}}$.

\setlength{\floatsep}{5pt}{
\begin{algorithm}[!ht]
\caption{\textbf{A novel update of a nominal stabilizing SOF controller}}\label{proced}
\DontPrintSemicolon
\textbf{Input}: $\Delta$, $B$, $C$, $F^{\text{nominal}}$.

Compute $G^{\ast}_{\Delta}$ via $G^{\ast}_{\Delta} = -B^{+} \Delta ({C}^{\top +})^\top$ in \eqref{Osol}.

Compute $F^{\text{updated}}$ via $F^{\text{updated}} = F^{\text{nominal}} + G^{\ast}_{\Delta}$.

\textbf{Output}: $F^{\text{updated}}$.

\end{algorithm}}

\subsection{On the computation of the MDRP}

On the one hand, as mentioned earlier in the paper, computing the exact value of MDRP $\beta$ is computationally intractable (NP-hard). On the other hand, the lower bound on the MDRP $\beta$ in \eqref{CSR} can be employed to provide a guaranteed stability region, but it is more conservative. Such limitations motivate us to use the following (heuristic) optimization problem:
\begin{align} \label{OPaux}
\underset{v \in \mathbb{R}^{n^2},~\beta \in \mathbb{R}_{++}}{\max}\alpha \bigg(A+BF^{\text{nominal}}C + \beta \mathbf{vec}^{-1}\bigg(\frac{v}{\|v\|}\bigg)\bigg)
\end{align}
along with a specialized bisection method (fixing the value of $\beta$ and solving for a $v \in \mathbb{R}^{n^2}$), to obtain a near-optimal value of $\beta$. The philosophy behind the construction of such a heuristic optimization problem \eqref{OPaux} is that by definition we have ($(3.2)$ in \cite{van1984near})
\begin{align} \label{betacomp}
    & \beta_{\mathbb{R}}(A+BF^{\text{nominal}}C) := \notag\\ 
    & \min \{\|\mathcal{X}\|_F: \alpha(A+BF^{\text{nominal}}C + \mathcal{X}) = 0, \mathcal{X} \in \mathbb{R}^{n \times n} \}
\end{align}or equivalently, for all $\mathcal{X} \in \mathbb{R}^{n \times n}$ satisfying the inequality
\begin{align*}
    & \|\mathcal{X}\|_F < \beta_{\mathbb{R}}(A+BF^{\text{nominal}}C)
\end{align*}
the inequality
\begin{align*}
    & \alpha(A+BF^{\text{nominal}}C + \mathcal{X}) < 0
\end{align*}
must be satisfied. Since one cannot check such a satisfaction for the infinitely many values of $\mathcal{X} \in \mathbb{R}^{n \times n}$ satisfying the inequality $\|\mathcal{X}\|_F < a$ for a guess value $a$, the satisfaction of the following inequality:
\begin{align*}
    \Big(\max_{\mathcal{X} \in \mathbb{R}^{n \times n},\|\mathcal{X}\|_F= b} \alpha(A+BF^{\text{nominal}}C + \mathcal{X})\Big) < 0.
\end{align*}should be checked where $b = a^{-}$ is a guess value tending to a guess value $a$ from below. In \eqref{OPaux}, with a bit of abuse of notation, we have utilized the notation $\beta$ instead of the notation $b$. Observe that $\|\mathcal{X}\|_F = \beta$ holds for the following parameterization:
\begin{align*}
    & \mathcal{X} = \beta \mathbf{vec}^{-1}\bigg(\frac{v}{\|v\|}\bigg)
\end{align*}
that we have chosen in \eqref{OPaux}.

We initialize $\beta$ with $\beta_{\mathbb{R}}^{u}$ and at each step, we check if the maximum value, namely $\alpha^{\ast}$, is non-negative or not (i.e., the updating criterion for the bisection method). To solve the optimization problem, one could utilize MATLAB's built-in function \texttt{fminunc(.)}. Since any optimization solver---including \texttt{fminunc(.)}---can fail at identifying the globally optimal solution of the heuristic optimization problem \eqref{OPaux}, the stability is not guaranteed when using this
heuristic, but the empirical results show that it works in practice. We emphasize that the efficiency guarantee of the proposed updated stabilizing SOF controller substantially depends on the computational efficiency of MDRP $\beta$ as the computational complexity of \eqref{Osol} is $\mathcal{O}(n^2 \min \{m,p\})$. Furthermore, it is noteworthy that solving the heuristic optimization problem \eqref{OPaux} along with a specialized bisection method to obtain a near-optimal value of $\beta$ can become challenging for higher-order systems as the dimension of $v$ is $n^2$ and quadratically increasing with $n$. Furthermore, while we do not provide any theoretical convergence proof for the proposed heuristic, we later on practically verify its effectiveness via extensive numerical case studies.

To approximately compute the MDRP $\beta$, we highlight that one can alternatively employ the HEC-based approximation method proposed by \cite{guglielmi2017approximating} that has been implemented in the robust stability package, namely ROSTAPACK, developed by \cite{Mitchell2022RObust}. Precisely, the corresponding command is called \texttt{getStabRadBound}. The exact mathematical equation to be approximately solved for $\beta$ via the HEC-based approximation method \cite{guglielmi2017approximating} is as follows:
\begin{align} \label{aAcl}
    & \alpha_{\beta}(A+BF^{\text{nominal}}C) = 0
\end{align}
where $\alpha_{\epsilon}(M)$ denotes the real $\epsilon$-pseudospectral abscissa of $M$ (also known as spectral value set abscissa), which is the largest of the real parts of the elements of the real $\epsilon$-pseudospectrum $\Lambda_{\epsilon}^{\mathbb{R}}(M)$, i.e.,
\begin{align*}
    & \alpha_{\epsilon}(M) := \max \{\Re(z): z \in \Lambda_{\epsilon}^{\mathbb{R}}(M)\}\\
    & \Lambda_{\epsilon}^{\mathbb{R}}(M) :=\\ 
    & \{\lambda \in \mathbb{C}: \lambda \in \Lambda(M+\epsilon E), E \in \mathbb{R}^{n \times n}, \|E\|_F \le 1\}.
\end{align*} Although $\alpha_{\beta}(A+BF^{\text{nominal}}C)$ is a monotonically increasing continuous function of $\beta$, enabling the usage of the Newton-bisection method, no reliable way exists to the exact evaluation of $\alpha_{\beta}(A+BF^{\text{nominal}}C)$ \cite{guglielmi2017approximating} which is aligned with the discussion previously provided in Section \ref{sec:Main}-A-3. It is noteworthy that \eqref{aAcl} is equivalent to
\begin{align} \label{epsE}
    & \alpha(A+BF^{\text{nominal}}C + \beta E(\beta)) = 0
\end{align}
where $\|E(\beta)\|_F = 1$ holds. Comparing the form of \eqref{epsE} with the form of \eqref{betacomp}, we conclude that the HEC-based approximation method proposed by \cite{guglielmi2017approximating} and our proposed heuristic are similar in nature.
We also highlight that the norm choice affects the RSR (MDRP) computation algorithms. The norm considered in this work and \cite{guglielmi2017approximating} is the Frobenius norm. However, the spectral norm has also been used in the literature \cite{rostami2015new,guglielmi2016method,lu2017criss} in this context. Therefore, to have a fair comparison, we choose the HEC-based approximation method proposed by \cite{guglielmi2017approximating} as a similar counterpart.

\subsection{A toy example}

Let us consider a lateral axis model of an $L-1011$ aircraft in cruise flight conditions ($AC3$) \cite{leibfritz2006compleib}. For such a model, we have $(n,m,p) = (5,2,4)$. We design the following nominal stabilizing SOF controller $F^{\text{nominal}}$ via MATLAB built-in function $\texttt{hinfstruct}(.)$:
\begin{align*}
F^{\text{nominal}} &= \begin{bmatrix} 0 & 0 & 0 & -0.5057 \\ 0.7521 & 0 & -3.0713 & 1.1408 \end{bmatrix}
\end{align*}
for which $\beta = 0.1931$ (computed via solving the optimization problem \eqref{OPaux}) and $\beta_{\mathbb{R}}^u = 0.3230$ (computed via \eqref{UpBoG}) hold. Moreover, using \eqref{CSR} along with the MATLAB built-in command $\texttt{hinfnorm}(.)$, we get the CSR $\beta_{\mathbb{R}}^l = 0.1322$. For comparison, we repeated the $\beta$ computation via the HEC-based approximation method proposed by \cite{guglielmi2017approximating} and obtained the same value $\beta^{\text{HEC}} = 0.1931$ as ours, i.e., $\beta = 0.1931$.

Fig. \ref{FigKS} (Left) visualizes the stability regions for $AC3$ benchmark with $\frac{\beta}{\rho} = \frac{1}{2}$: the guaranteed (conservative) stability region based on Proposition \ref{Prop2} and the exact one based on $\alpha(A+ BF^{\text{nominal}}C + \Delta + B G_{\Delta}^{\ast}C) < 0$ with $G_{\Delta}^{\ast}$ in \eqref{Osol}. As expected, the guaranteed (conservative) stability region is a subset of the exact one.

For instance, the update $G_{\Delta}^{\ast}$ for a destabilizing $\Delta$ (i.e., $\alpha(A + BF^{\text{nominal}}C + \Delta) = 0.0483 \nless 0$)
\begin{align*}
    & \Delta =\\ 
    & \begin{bmatrix}
        0.0324  & -0.0120  & -0.0130 &   0.0186  & -0.0103\\
        -0.0681  &  0.0949  & -0.1547 &   0.0178  & -0.0175\\
        0.0126  & -0.0203  &  0.0348 &  -0.0052  &  0.0042\\
        0.0897 &  -0.0063 &  -0.0858  &  0.0672 &  -0.0312\\
        -0.0070  &  0.0010  &  0.0059  & -0.0050  &  0.0024
    \end{bmatrix}
\end{align*}
with $(\tau_{\Delta},\theta_{\Delta}) = (0.45,0.45)$ can be computed via Procedure \ref{proced} as follows:
\begin{align*}
G_{\Delta}^{\ast} &= \begin{bmatrix} 0.0745 & -0.2034 & 0.0214 & -0.0939 \\ 0.0115 & -0.0302 & 0.0018 & -0.0169 \end{bmatrix}
\end{align*}
for which $\|BG_{\Delta}^{\ast}C + \Delta\|_F = 0.1629 < \beta = 0.1931$ and $\alpha(A+ BF^{\text{nominal}}C + \Delta + B G_{\Delta}^{\ast}C) = -0.0637 < 0$ hold and the updated stabilizing SOF controller $F^{\text{updated}} = F^{\text{nominal}} + G_{\Delta}^{\ast}$ is as follows:
\begin{align*}
F^{\text{updated}} &= \begin{bmatrix} 0.0745 & -0.2034 & 0.0214 & -0.5996 \\ 0.7636 & -0.0302 & -3.0695 & 1.1239 \end{bmatrix}.
\end{align*}
\begin{figure}[t]
\subfloat[a][]{\includegraphics[scale=0.24]{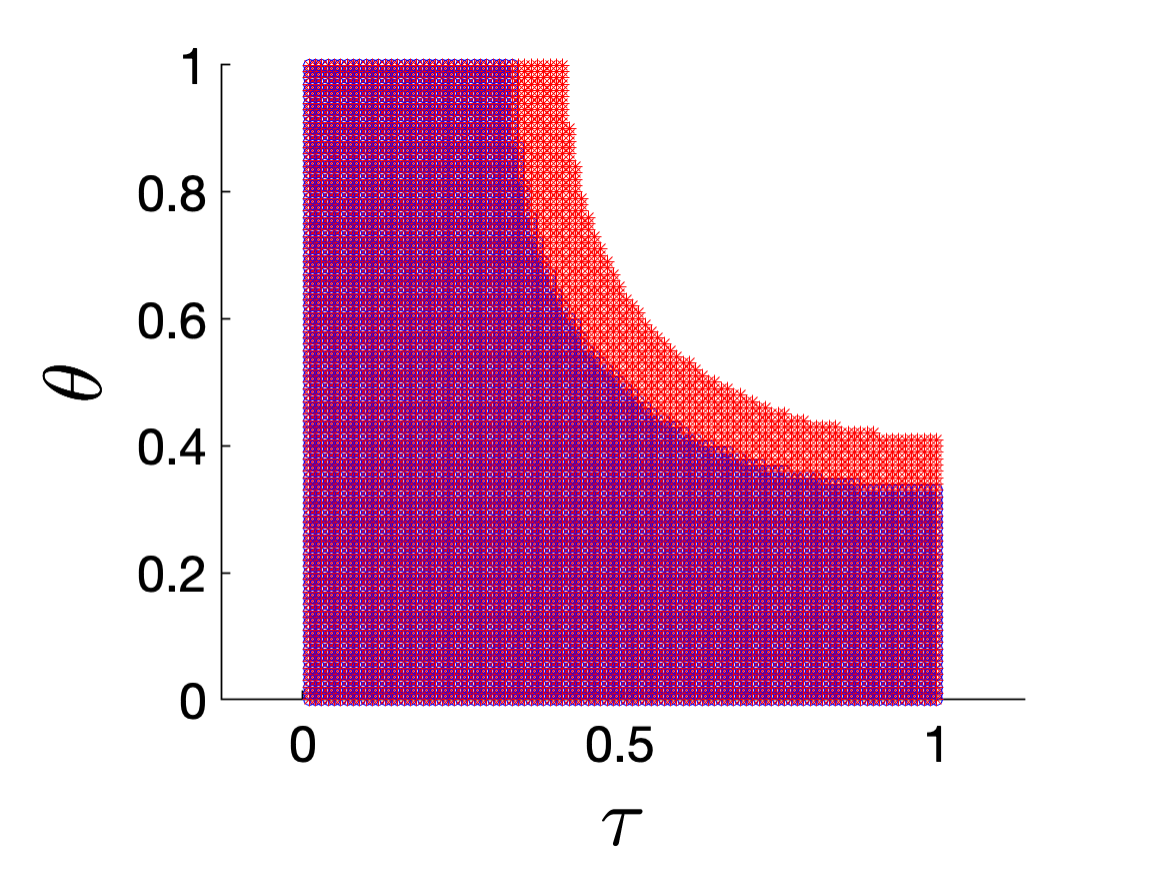}\label{<figu1>}}
\subfloat[b][]{\includegraphics[scale=0.24]{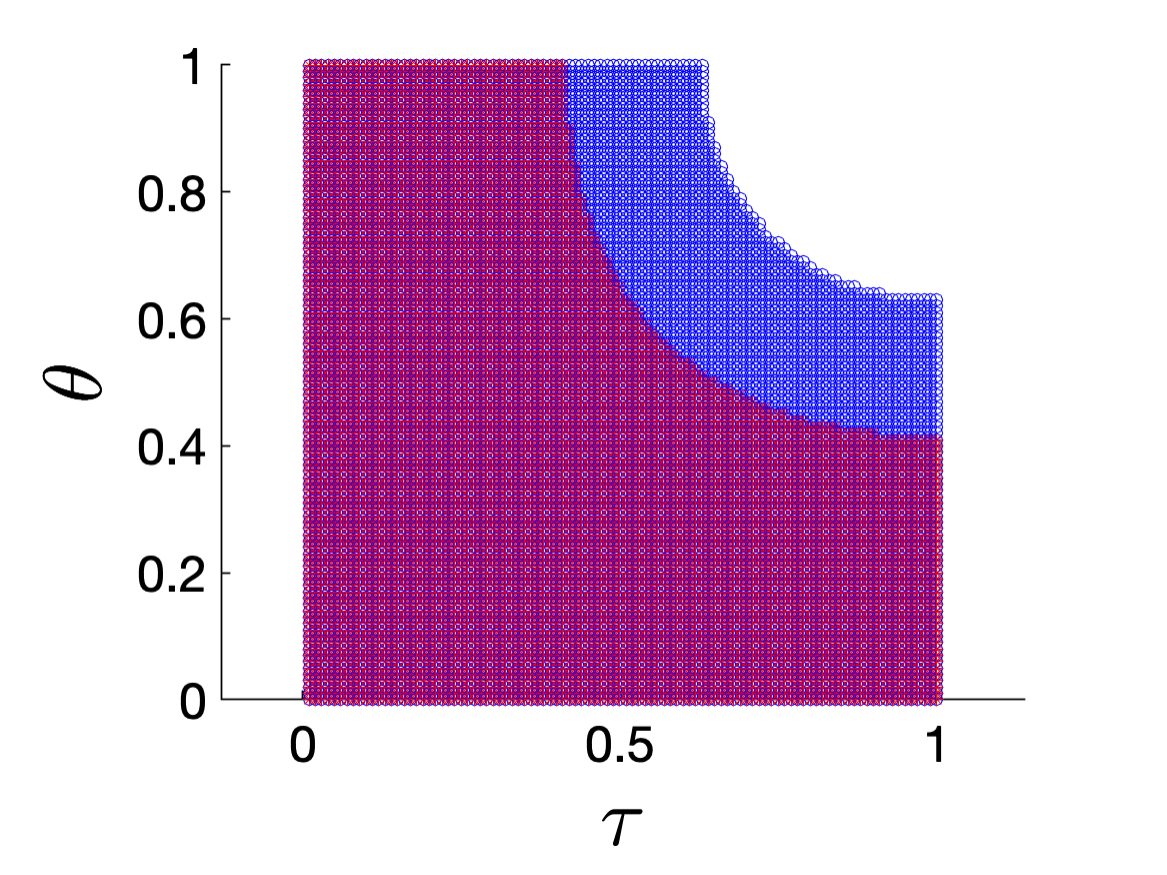}\label{<figu2>}}

\caption{Stability regions for $AC3$ benchmark with (a) $\rho = 2\beta^{\mathrm{accurate}}$ and $\beta = \beta^{\mathrm{accurate}}$ (computed via solving the optimization problem \eqref{OPaux}) and (b) $\rho = 2\beta^{\mathrm{accurate}}$ and $\beta = \beta^{\mathrm{inaccurate}}$: the guaranteed (conservative) stability regions based on Proposition \ref{Prop2} (filled with \textcolor{black}{blue} circles) and the exact ones based on $\alpha(A+ BF^{\text{nominal}}C + \Delta + B G_{\Delta}^{\ast}C) = \alpha(A+\Delta + B F^{\text{updated}} C) < 0$ with $G_{\Delta}^{\ast}$ in \eqref{Osol} (filled with \textcolor{black}{red} asterisks).} \label{FigKS}
\end{figure}

Remarkably, the accurate computing of $\beta$ via solving the optimization problem \eqref{OPaux} plays a significant role in accurately identifying the stability regions. As Fig. \ref{FigKS} (Right) depicts, choosing $\rho$ equal to $2 \times 0.1931$ (as chosen for Fig. \ref{FigKS} (Left)) and $\beta$ equal to $0.3230$ (an inaccurate value), leads to the misleading stability regions. First, the guaranteed (conservative) stability region has erroneously been enlarged. Second, the guaranteed (conservative) stability region has erroneously become the superset of the exact one.

For comparison, we added randomly-generated changes (via MATLAB built-in function $\texttt{rand}(.)$) to the state-space matrices associated with $AC3$ \cite{leibfritz2006compleib} to get the following randomly-generated matrices $(A,B,C)$ along with the computed nominal stabilizing SOF controller $F^{\text{nominal}}$ via MATLAB built-in function $\texttt{hinfstruct}(.)$: 
\begin{align*}
    & A =\\ 
    & \begin{bmatrix} 0.3568  &  0.3436  &  1.2383  &  0.1679  &  0.3448\\
    0.2808  & -0.2208  &  0.0488  &  1.4823 &  -0.2322\\
   -0.2263 &  -0.2345  & -0.9835  & -4.8075  & -0.0586\\
    0.3235  & -0.9357  &  0.1612 &  -0.2072  & -0.2350\\
    0.4205  &  0.4980  & -0.4875  & -0.2943  & -0.1271
    \end{bmatrix}\\
    & B = \begin{bmatrix}
        0.4356  & -0.0332\\
   -1.2120  & -0.1706\\
    0.4148  & -1.3938\\
    0.4407  & -0.4383\\
    0.4842  &  0.1401
    \end{bmatrix}\\
    & C =\\
    & \begin{bmatrix}
         0.4943  &  0.5142 &  -0.3319  &  0.4558  & -1.4760\\
   -0.2771  &  0.1801  &  0.6128  &  0.3849  &  0.3282\\
   -0.2527  &  0.1640  & -0.0768  &  1.4660  &  0.0418\\
    1.4222  &  0.1010  & -0.2804 &  -0.1668  &  0.1605
    \end{bmatrix}\\
    & F^{\text{nominal}} = \begin{bmatrix} 0 & 0 & 0 & 0.1919 \\ 0.1571 & 0 & -0.6671 & 0.2497 \end{bmatrix}
\end{align*}
and repeated the simulations for both our heuristic optimization method and the HEC-based approximation method proposed by \cite{guglielmi2017approximating}. According to \eqref{UpBoG}, we get $\beta_{\mathbb{R}}^u = 0.0320$. Our heuristic optimization method gives: $\beta = 0.0320$ which is equal to the CSR $\beta_{\mathbb{R}}^l = 0.0320$. However, the HEC-based approximation method \cite{guglielmi2017approximating} gives: $\beta^{\text{HEC}} = 0.0586$. Thus, we observed that ours outperformed theirs. Indeed, $\beta = 0.0586$ is a loose upper bound on the MDRP rather than the exact value. To prove that, we used the optimization problem and solved it for $v$, setting the value of $\beta$ equal to the fixed value of $\beta^{\text{HEC}} = 0.0586$ and obtained the following $\mathcal{X}(\beta^{\text{HEC}})$:
\begin{align*}
    & \beta^{\text{HEC}} \mathbf{vec}^{-1}\bigg(\frac{v}{\|v\|}\bigg) =\\
    & \begin{bmatrix}
        0.0130  &  0.0008  &  0.0015 &  -0.0018  & -0.0079\\
    0.0053  &  0.0087  &  0.0274  & -0.0086 &  -0.0383\\
    0.0018  &  0.0030  &  0.0096 &  -0.0030 &  -0.0134\\
   -0.0187  &  0.0010  &  0.0049  &  0.0006  &  0.0022\\
    0.0068  &  0.0013  &  0.0036  & -0.0018  & -0.0078
    \end{bmatrix}
\end{align*}
for which $\alpha(A+BF^{\text{nominal}}C + \mathcal{X}(\beta^{\text{HEC}})) = 0.0187 \nleq 0$ holds. Thus, according to \eqref{aAcl}, the exact (true) value of $\beta$ is less than $\beta^{\text{HEC}} = 0.0586$.

\subsection{Two power system examples}

\subsubsection{Multi-machine power systems}

Here, we demonstrate how the quick updates for the SOF control law can stabilize multi-machine power networks while still performing nearly identically to the more computationally intensive SOF formulations. In particular, we consider a multi-machine power system represented by the following nonlinear differential-algebraic equation (NDAE) model~\cite{kundur2007power,sauer2017power}:
\begin{subequations} \label{NDAEs}
    \begin{align}
    \dot{x}(t) & = g(x(t),a(t),u(t))\\
    d(t) & = h(x(t),a(t))
\end{align}
\end{subequations}where $x(t) \in \mathbb{R}^{n_sG}$, $a(t) \in \mathbb{R}^{2N+2G}$, and $u(t) \in \mathbb{R}^{n_cG}$ denote the differential states vector (frequencies, angles, emfs, et cetera), algebraic states vector (bus voltages and currents), and control input vector (excitation field voltage and mechanical input power), respectively with $N$, $G$, $n_s$, $n_c$ as number of buses, number of generators, number of states and control inputs in each generator. Vector $d(t)$ models demand (loads) in the $n_d$ buses in the network. The vector-valued mapping $g(\cdot)$ captures the generator physics, and the mapping $h(\cdot)$ models the algebraic power flow equations (i.e., KCL/KVL in the power network). This NDAE model captures the fundamental electromechanical dynamics of the grid and is widely used in the industry to predict power system transients and to perform feedback control and state estimation. 

A large volume of feedback control architectures in power systems resorts to linearizing the NDAE-modeled dynamics \eqref{NDAEs} around a known equilibrium point $z^{0} = (x^0,a^0,u^0)$. The reader is referred to~\cite{sadamoto2019dynamic, mondal2020power} for some thorough expositions on feedback control and modeling of power systems, and linearized, small-signal power system models. After computing the equilibrium point and eliminating the linearized algebraic equation, one can get the following linear ordinary differential equation (ODE) state-space representation:
\begin{align}
    \Delta \dot{x}'(t) & = A(z^0) \Delta x'(t) + B(z^0) \Delta u'(t)
\end{align}
with $\Delta x'(t) := \Delta x(t) - \Delta x^{s}(t)$, $\Delta x(t) = x(t) - x^{0}$, and $\Delta x'(0) = x^{0} - x^{s}$ where superscript $^{s}$ denotes the desired equilibrium point. Also, $A(z^0)$ and $B(z^0)$ are computed as follows:
\begin{align*}
    A(z^0) &= g_x(z^0) - g_a(z^0) h_a^{-1}(x^0,a^0)h_x(x^0,a^0)\\
    B(z^0) &= g_u(z^0) 
\end{align*}
where $g_x$, $g_a$, $h_x$, and $h_a$ represent the Jacobian matrices associated with $g$ and $h$ with respect to $x$ and $a$, respectively. Initially, the power system is operating at a known equilibrium point $z^{0}$. By solving an LQR-based optimal power flow (OPF) problem (LQR-OPF problem in \cite{bazr2017,bazrafshan2018coupling}, or virtually any other feedback controller), one can get a nominal LQR state feedback controller $K^{\text{LQR,nominal}}$ that steers the system to the desired equilibrium point, namely $z^s$. LQR-OPF problem, which is an SDP consisting of the LMIs, is periodically solved for the desired equilibrium point after each period (typically in minutes). We will use such $K^{\text{LQR,nominal}}$ and solve $K^{\text{LQR,nominal}} = F^{\text{nominal}}C$ for a nominal stabilizing SOF controller $F^{\text{nominal}}$ via $F^{\text{nominal}} = K^{\text{LQR,nominal}}C^\top (C C^\top)^{-1}$ \cite{gadewadikar2009h} given an output matrix $C$ as $C = \texttt{eye}(p,n)$ satisfying $p < n$. Moreover, as a known perturbation $\Delta$, we have $\Delta = A(z^s) - A(z^0)$ wherein
\begin{align*}
    A(z^s) &= g_x(z^s) - g_a(z^s) h_a^{-1}(x^s,a^s)h_x(x^s,a^s)\\
    B(z^s) &= g_u(z^s).
\end{align*}
This procedure (of obtaining a new operating point and computing the updated feedback gain) is repeated in two different power scenarios or contexts. The first pertains to changing optimal power flow (OPF) setpoints, while the second can take place when a fault happens in the system, resulting in explicit changes to the state-space matrices $A$ and $B$. Both scenarios yield updated $\Delta$ or a perturbed system model. 

\begin{figure}[t]
	\centering
	\includegraphics[scale=0.4]{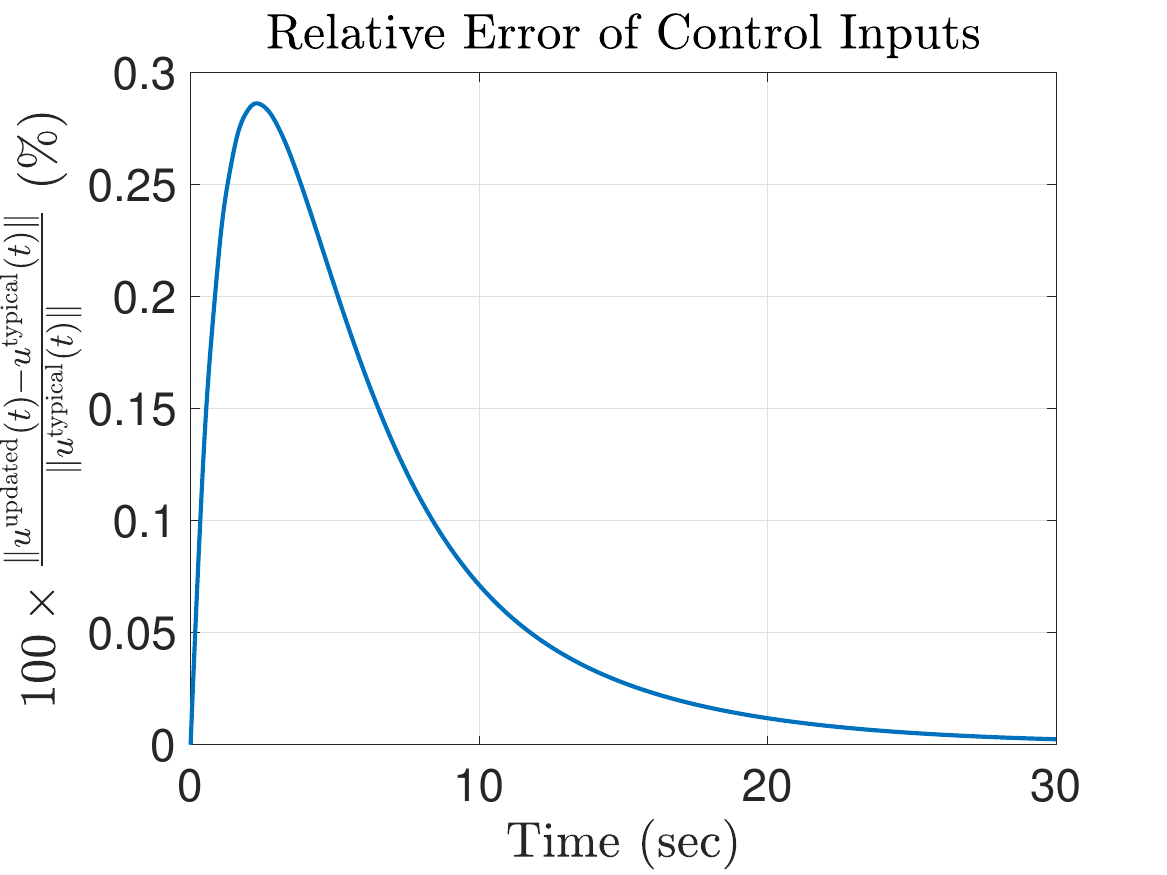}
	\caption{Relative error of control inputs $100 \times \frac{\|u^{\text{updated}}(t)-u^{\text{typical}}(t)\|}{\|u^{\text{typical}}(t)\|}$ ---with $u^{\text{updated}}(t) = F^{\text{updated}} y^{\text{updated}}(t)$ and $u^{\text{typical}}(t) = F^{\text{typical}} y^{\text{typical}}(t)$ as control inputs considering case $14$-bus.}
	\label{fig6}
\end{figure}

\begin{table}[t]
	\centering
	\caption{Computational time for two approaches: (\textit{i}) typical (SDP-based alternative \cite{bazrafshan2018coupling}) and (\textit{ii}) updating (Procedure \ref{proced}), considering the three IEEE power system test cases: case $9$-bus, case $14$-bus, and case $39$-bus.}
	{\small \begin{tabular}{|c|c|c|}
			\hline
			\diagbox{Case}{Approach} & Typical \cite{bazrafshan2018coupling} & Updating (Procedure \ref{proced}) \\
			\hline
			$9$-bus & $1.1444$ s & \cellcolor{lightgray} $0.0048$ s\\
			\hline
			$14$-bus & $3.6817$ s & \cellcolor{lightgray} $0.0077$ s \\
			\hline
			$39$-bus & $9.6370$ s & \cellcolor{lightgray} $0.0079$ s\\
			\hline
		\end{tabular}
		\label{tab0}}
\end{table}

\begin{figure*}[t]
	\centering
	
	\subfloat[a][]{\includegraphics[scale=0.3]{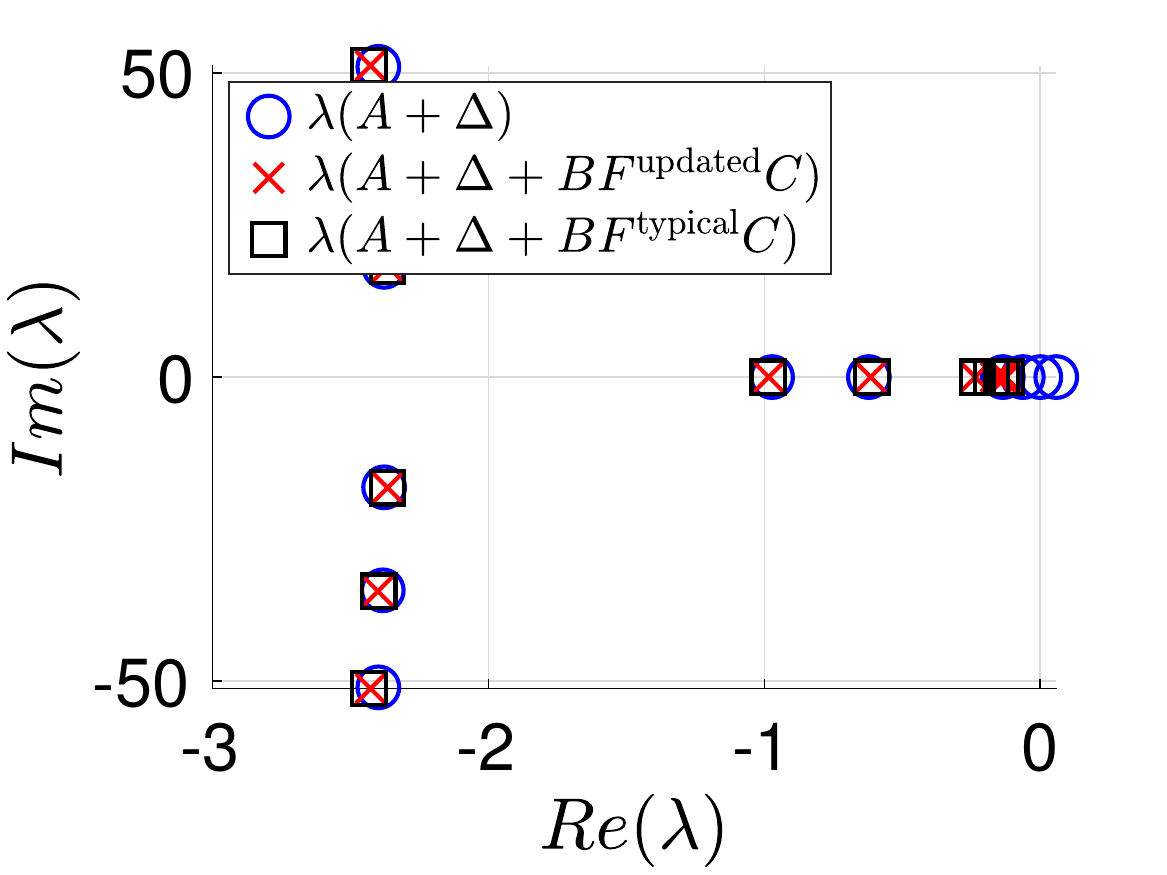}\label{<fig1>}}
	\subfloat[b][]{\includegraphics[scale=0.3]{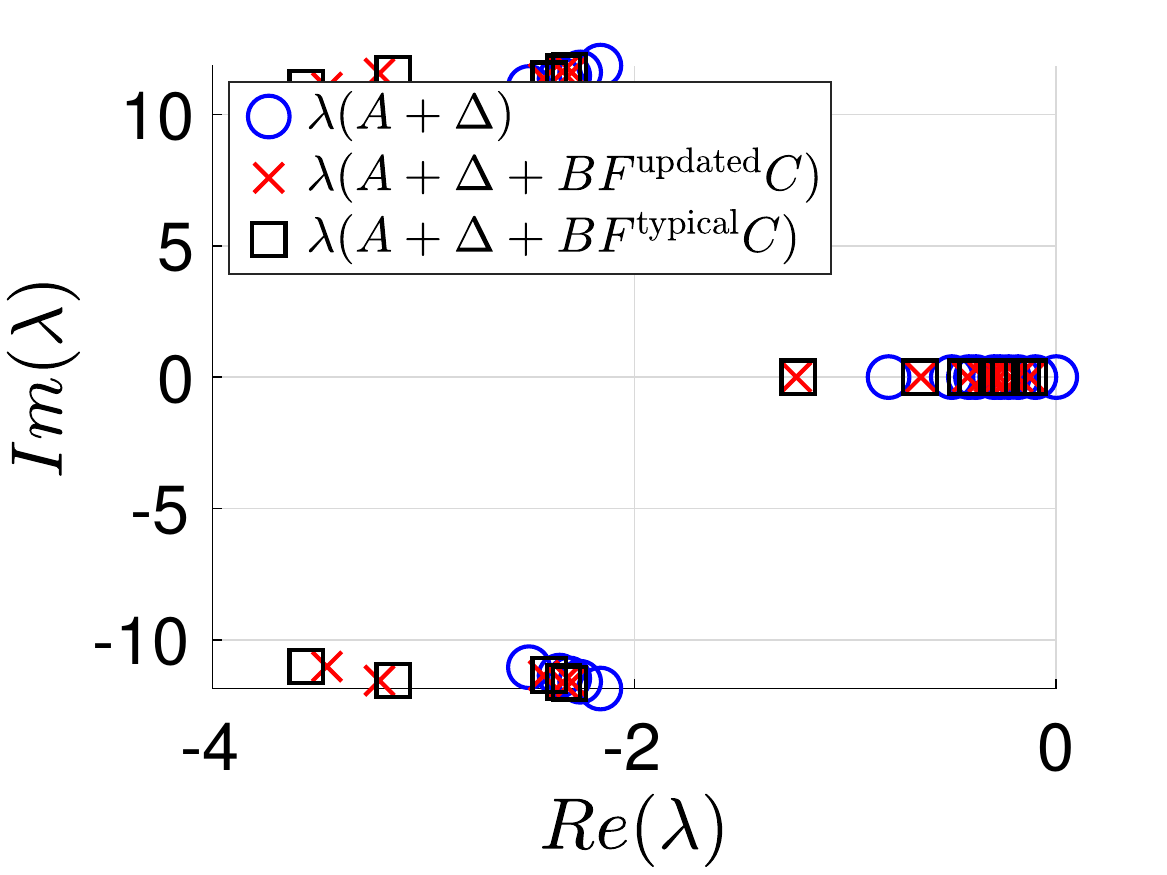}\label{<fig2>}}
	\subfloat[c][]{\includegraphics[scale=0.3]{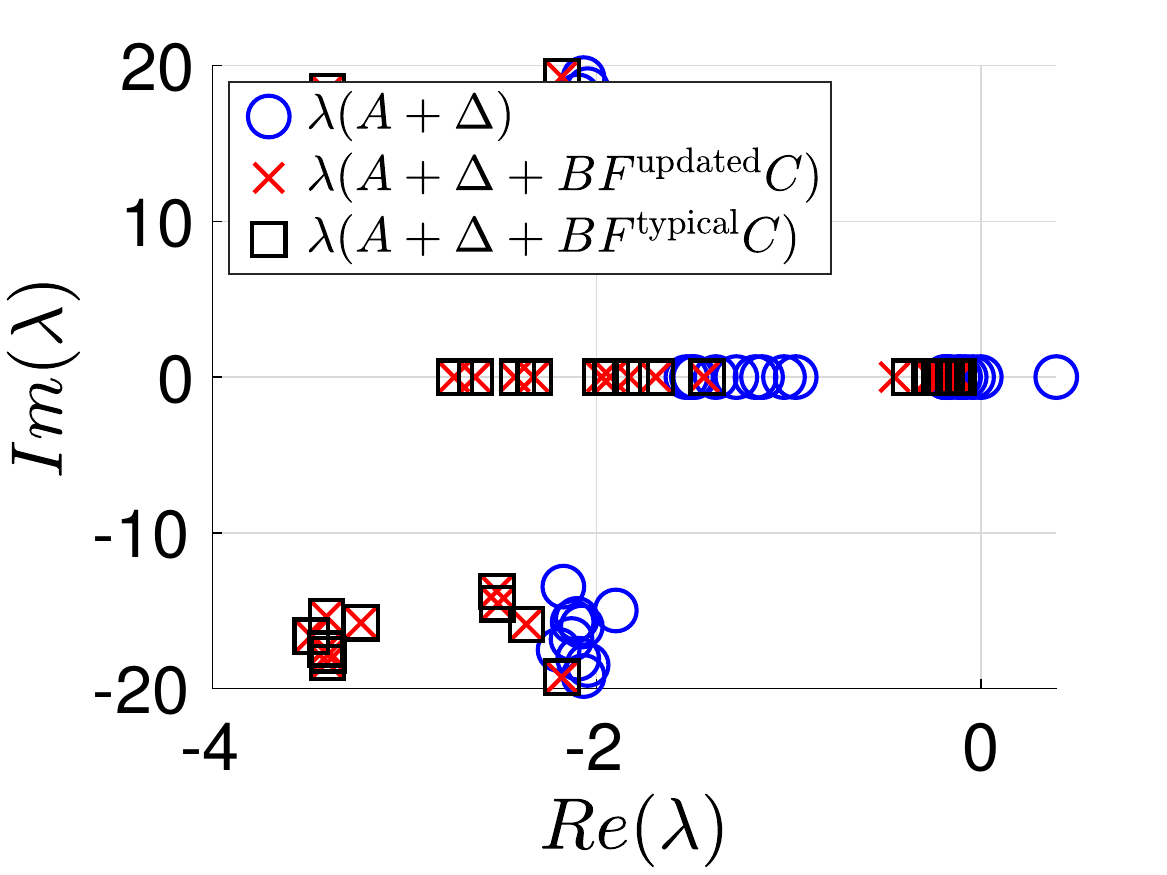}\label{<fig3>}}
	
	\caption{Spectra associated with $A+\Delta$, $A+\Delta + BF^{\text{updated}}C$, and $A+\Delta + BF^{\text{typical}}C$ considering the three IEEE power system test cases: (a) case $9$-bus, (b) case $14$-bus, and (c) case $39$-bus.}
	\label{fig:eigs}
\end{figure*}

Fig. \ref{fig6} visualizes the relative error of control inputs $100 \times \frac{\|u^{\text{updated}}(t)-u^{\text{typical}}(t)\|}{\|u^{\text{typical}}(t)\|}$---associated with $u^{\text{updated}}(t) = F^{\text{updated}} y^{\text{updated}}(t)$ and $u^{\text{typical}}(t) = F^{\text{typical}} y^{\text{typical}}(t)$ considering case $14$-bus, respectively. As Fig. \ref{fig6} depicts, there is no meaningful change (a maximum relative error of $\approx 0.3\%$) in the control law, i.e., this will incur almost no additional cost.

Considering the data information of three IEEE power system test cases (each generator has been modeled with a $4$-th order dynamics): case $9$-bus with $(N,G,n_s,n_c) = (9,3,4,2)$, case $14$-bus with $(N,G,n_s,n_c) = (14,5,4,2)$, and case $39$-bus with $(N,G,n_s,n_c) = (39,10,4,2)$, we compare the computational time for the following two approaches: (\textit{i}) repeating the whole SDP-based design procedure LQR-OPF and computing $F^{\text{typical}}$ via solving LQR-OPF for $K^{\text{LQR,typical}}$ along with $F^{\text{typical}} = K^{\text{LQR,typical}}C^\top (C C^\top)^{-1}$ and (\textit{ii}) computing the quick update $G_{\Delta}^{\ast}$ via \eqref{Osol}, i.e., $G_{\Delta}^{\ast} = -B^{+} \Delta ({C}^{\top +})^\top$ and getting $F^{\text{updated}} = K^{\text{LQR,nominal}}C^\top (C C^\top)^{-1} -B^{+} \Delta ({C}^{\top +})^\top$ via Procedure \ref{proced}. We choose $C$ as $C = \texttt{eye}(p,n)$ with $p = n-3$. Moreover, the relative norm of the perturbation defined as $100 \times \frac{\|\Delta\|_F}{\|A(z^0)\|_F}~(\%)$ takes the following values for the three IEEE power system test cases: case $9$-bus, case $14$-bus, and case $39$-bus, respectively: $4.0044\%$, $1.4886\%$, and $8.1370\%$. Tab.~\ref{tab0} presents the computational time for two approaches: (\textit{i}) typical (SDP-based alternative \cite{bazrafshan2018coupling}) and (\textit{ii}) updating (Procedure \ref{proced}), considering the three IEEE power system test cases: case $9$-bus, case $14$-bus, and case $39$-bus. It is noteworthy that for a large-scale power system test case, namely case $200$-bus with $196$ states, we waited more than $12$ hours and still did not get a solution from the SDP-based design procedure LQR-OPF. Subject to the availability of $F^{\text{nominal}}$ and $\Delta$, our proposed quick update could be computed in less than a second. 

As Tab.~\ref{tab0} reflects, the analytical updating approach outperforms the typical SDP-based alternative in terms of computational time. We highlight that for the three IEEE power system test cases: case $9$-bus, case $14$-bus, and case $39$-bus, $\alpha(A+ BF^{\text{nominal}}C + \Delta + B G_{\Delta}^{\ast}C)$ takes the following values, respectively: $-0.1253 < 0$, $-0.1256 < 0$, and $-0.1169 < 0$ confirming the success of the proposed updating approach in terms of closed-loop stabilization subject to a known perturbation. Fig. \ref{fig:eigs} visualizes the spectra associated with $A+\Delta$, $A+\Delta + BF^{\text{updated}}C$, and $A+\Delta + BF^{\text{typical}}C$ considering the three IEEE power system test cases: case $9$-bus, case $14$-bus, and case $39$-bus, respectively. As Fig. \ref{fig:eigs} depicts, the proposed updating approach can successfully attenuate the effect of the destabilizing real perturbation $\Delta$ and attain a stabilization performance similar to the typical SDP-based alternative while significantly reducing the computational time. Fig. \ref{fig:stats} visualizes some of the state trajectories---generator voltage magnitudes, generator voltage phases, real powers, and reactive powers---associated with $A+\Delta + BF^{\text{updated}}C$ (Left) and $A+\Delta + BF^{\text{typical}}C$ (Right) considering case $14$-bus, respectively. As Fig. \ref{fig:stats} depicts, the proposed updating approach can successfully attenuate the effect of the destabilizing real perturbation $\Delta$ and attain a stabilization performance similar to the typical SDP-based alternative while significantly reducing the computational time.

\subsubsection{Second-order swing dynamics-based power systems}

We consider the case of power systems with the second-order swing dynamics \cite{kundur1994power,wu2014sparsity}
\begin{align*}
    & \dot{x} = A(L)x + Bu,~A(L)=T^\top \bar{A}(L) T,B = T^\top \bar{B}\\
    & \bar{A}(L) = \begin{bmatrix}
        0 & I_N\\-M^{-1}L & -M^{-1}D
    \end{bmatrix}, \bar{B} = \begin{bmatrix}
            0\\M^{-1}
        \end{bmatrix}\\
    & T = \begin{bmatrix}
            U & 0\\0 & I_N
        \end{bmatrix}, U^\top U = I_{N-1}, UU^\top = I_N - \frac{1}{N} \mathbf{1}\mathbf{1}^\top\\ 
        & U^\top \mathbf{1}_N = 0_{N-1},~N: \text{Number~of~Generators}
\end{align*}
where $M \in \mathbb{R}^{N \times N}$, $D \in \mathbb{R}^{N \times N}$, and $L \in \mathbb{R}^{N \times N}$ denote the diagonal matrix of generator inertia coefficients, diagonal matrix of generator damping coefficients, and Laplacian matrix of susceptance values among the generators, respectively.

Note that we can implement the susceptance link removals via perturbing $L$. Observe that we have
\begin{align*}
    & \bar{A}(L+\delta L) = \begin{bmatrix}
        0 & I_N\\-M^{-1}(L+\delta L) & -M^{-1}D
    \end{bmatrix} =\\ 
    & \bar{A}(L) + \begin{bmatrix}
        0 & 0\\-M^{-1}(\delta L) & 0
    \end{bmatrix}\\
    & A(L+\delta L) = T^{\top} \bar{A}(L+\delta L)T= \\
    & A(L) + \underbrace{T^\top\begin{bmatrix}
        0 & 0\\-M^{-1}(\delta L) & 0
    \end{bmatrix}T}_{\Delta}. 
\end{align*}
Then, we utilize the following criterion (the main idea behind Procedure \ref{proced}):
\begin{align*}
    & \Bigg \|BG^{\ast}_{\Delta}C + \underbrace{T^\top\begin{bmatrix}
        0 & 0\\-M^{-1}(\delta L) & 0
    \end{bmatrix}T}_{\Delta} \Bigg \|_F \le\\
    & \underbrace{\beta_{\mathbb{R}}(A(L) + BF^{\text{nominal}}C)}_{\beta}
\end{align*}
where $G^{\ast}_{\Delta} = -B^{+} \Delta ({C}^{\top +})^\top$ according to \eqref{Osol}.

Considering the second-order swing dynamics for the IEEE case $14$-bus in \cite{11010910} and choosing $C = \texttt{eye}(N,2N-1)$, we get $F^{\text{nominal}} = K^{\text{LQR,nominal}}C^\top (C C^\top)^{-1}$ as follows:
\begin{align*}
    & \begin{bmatrix}
        0.0750  & -0.0132  &  0.0259  &  0.1828 &   -0.9492\\
   -0.0326  & -0.0201  &  0.0502  &  0.1593 &  -0.0104\\
   -0.0105  & -0.0871  &  0.0303  &  0.0108 &   0.0006\\
   -0.1247  &  0.0525  &  0.2829  &  0.0232 &   0.0044\\
   -0.1955  &  0.1139  &  0.2038 &   0.0573 &    0.0060
    \end{bmatrix}
\end{align*}
for which $\beta = 0.0917$ (computed via solving the optimization problem \eqref{OPaux}) and $\beta_{\mathbb{R}}^u = 0.5750$ (computed via \eqref{UpBoG}) hold. Also, for the CSR, $\beta_{\mathbb{R}}^l = 0.0442$ holds according to \eqref{CSR}. For instance, the update $G_{\Delta}^{\ast}$ associated with the removals of the susceptance links $(3,4)$ and $(1,5)$ (identified as the two most influential links in \cite{11010910} via an edge centrality matrix (ECM))
\begin{align*}
    & \Delta = T^\top\begin{bmatrix}
        0 & 0\\-M^{-1}(\delta L) & 0
    \end{bmatrix}T,~T = \begin{bmatrix}
            U & 0\\0 & I_N
        \end{bmatrix}\\
        & U = \begin{bmatrix}
            -0.4714 &  -0.1667  & -0.2236  &  0.7071\\
   -0.4714  & -0.1667  & -0.2236  & -0.7071\\
    0.2357 &   0.8333  & -0.2236  &  0.0000\\
    0.7071  & -0.5000 &  -0.2236  &  0.0000\\
         0    &     0  &  0.8944    &     0
        \end{bmatrix}\\
        & M = \begin{bmatrix}
            0.1273    &     0    &     0    &     0 &        0\\
         0   & 0.0541     &    0     &    0   &      0\\
         0    &     0  &  0.0308  &       0  &       0\\
         0    &     0     &    0  &  0.0775  &       0\\
         0   &      0    &     0    &     0 &    0.1379
        \end{bmatrix}\\
        & \delta L = \begin{bmatrix}
            -4.0759    &     0   &      0   &      0   & 4.0759\\
         0     &    0     &    0     &    0 &        0\\
         0     &    0  & -5.3693  &  5.3693 &        0\\
         0    &     0  &  5.3693 &  -5.3693 &        0\\
    4.0759   &      0     &    0    &     0 &   -4.0759
        \end{bmatrix}
\end{align*}
can be computed via Procedure \ref{proced} as follows:
\begin{align*}
    & G_{\Delta}^{\ast} = \begin{bmatrix}
        1.9214 & 0.6793 & 4.5570 & -2.8821 & 0\\
        0 & 0 & 0 & 0 & 0\\
        2.5311 & -7.1591 & 0.0000 & 0.0000 & 0\\
        -2.5311 & 7.1591 & -0.0000 & -0.0000 & 0\\
        -1.9214 & -0.6793 & -4.5570 & 2.8821 & 0
    \end{bmatrix}
\end{align*}
for which $\|B G_{\Delta}^{\ast}C + \Delta\|_F = 0 < \beta = 0.0917$ and $\alpha(A(L)+ BF^{\text{nominal}}C + \Delta + B G_{\Delta}^{\ast}C) = -0.6092 < 0$ hold and the updated stabilizing SOF controller $F^{\text{updated}} = F^{\text{nominal}} + G_{\Delta}^{\ast}$ is as follows:
\begin{align*}
& \begin{bmatrix} 1.9964   & 0.6661  &  4.5829  & -2.6993 &  -0.9492 \\ -0.0326  & -0.0201  &  0.0502  &  0.1593   & -0.0104\\
2.5206  & -7.2461  &  0.0303  &  0.0108  &  0.0006\\
-2.6558  &  7.2116  &  0.2829  &  0.0232  &  0.0044\\
-2.1169  & -0.5655  & -4.3531  &  2.9394 &   0.0060
\end{bmatrix}.
\end{align*}Since $\alpha(A(L)+ BF^{\text{nominal}}C + \Delta + B G_{\Delta}^{\ast}C) = -0.6092 < 0$ holds, the updated SOF controller $F^{\text{updated}} = F^{\text{nominal}} + G_{\Delta}^{\ast}$ computed by the algorithm, has successfully stabilized the perturbed  $A(L)$, i.e., $A(L) + \Delta$. For comparison, we repeated the $\beta$ computation via the HEC-based approximation method proposed by \cite{guglielmi2017approximating} and obtained $\beta^{\text{HEC}} = 0.2231$. Thus, we observed that ours outperformed theirs. Indeed, $\beta = 0.2231$ is a loose upper bound on the MDRP rather than the exact value. To prove that, we used the optimization problem and solved it for $v$, setting the value of $\beta$ equal to the fixed value of $\beta^{\text{HEC}} = 0.2231$ and obtained the following $\mathcal{X}(\beta^{\text{HEC}})$:
\begin{widetext}
\[
\beta^{\text{HEC}} \mathbf{vec}^{-1}\bigg(\frac{v}{\|v\|}\bigg) = 
    \begin{bmatrix} 0.0023  & -0.0021 &  -0.0009 &    0.0015  & -0.0170  &  0.0757& -0.0411  &  0.0299 &  -0.0225\\
   -0.0014  &  0.0031  &  0.0016  &  0.0004  &  0.0161 &  -0.0782 & 0.0080  &  0.0267  & -0.0023\\
   -0.0020  &  0.0039  &  0.0020  &  0.0002  &  0.0214  & -0.1025& 0.0165 &   0.0249  &  0.0015\\
    0.0025  &  0.0009  &  0.0009  &  0.0040 &  -0.0081  &  0.0248& -0.0747  &  0.1117 &  -0.0527\\
   -0.0000  &  0.0001  &  0.0000  &  0.0000  &  0.0004 &  -0.0019& -0.0001  &  0.0012  & -0.0003\\
    0.0001 &  -0.0002 &  -0.0001  & -0.0000  & -0.0009   & 0.0043& -0.0010  & -0.0005  & -0.0003\\
   -0.0000   & 0.0000  &  0.0000 &  -0.0000  &  0.0002 &  -0.0008 &0.0008  & -0.0010 &   0.0005\\
    0.0000  &  0.0001   & 0.0000  &  0.0002 &   0.0000  & -0.0006& -0.0023  &  0.0043 &  -0.0018\\
   -0.0001 &  -0.0000  & -0.0000 &  -0.0001  &  0.0002  & -0.0006 & 0.0022 &  -0.0033  &  0.0016
\end{bmatrix}\]
\end{widetext}\hspace{-0.125in} for which $\alpha(A+BF^{\text{nominal}}C + \mathcal{X}(\beta^{\text{HEC}})) = 1.2874 \nleq 0$ holds. Thus, according to \eqref{aAcl}, the exact (true) value of $\beta$ is less than $\beta^{\text{HEC}} = 0.2231$.

Remarkably, the MDRP in this power system is conservative as it is computed considering all Frobenius norm-bounded unstructured perturbations. In other terms, to obtain the specialized exact MDRP in this power system, the HEC-based approximation method and our heuristic method should be modified accordingly to take into account the specific structure associated with the structured perturbation $T^\top \begin{bmatrix}
    0 & 0\\-M^{-1}(\delta L) & 0
\end{bmatrix}$. Such a modification to the corresponding methods is eliminated due to space limitations and left as future work. Note that the aforementioned conservatism originates from the fact that the set of structured perturbations is a strict subset of the set of unstructured ones.

\begin{figure}[t]
	\centering
	
	\subfloat[a][]{\includegraphics[scale=0.225]{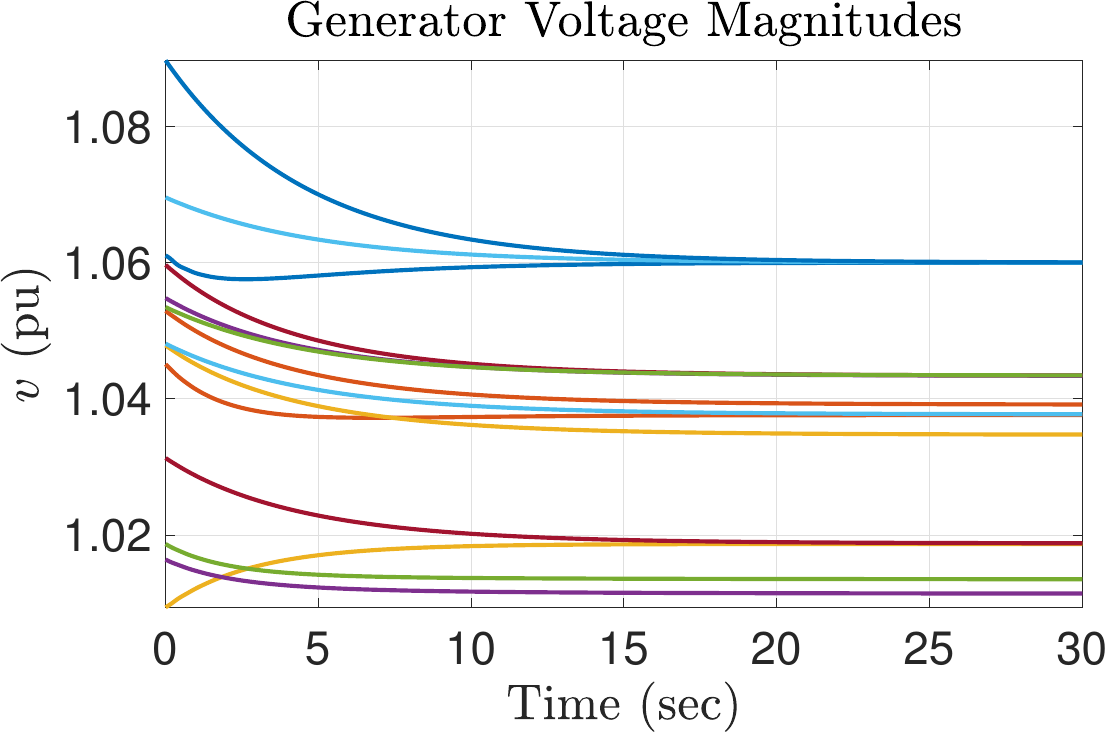}\label{<fi1>}}
	\subfloat[b][]{\includegraphics[scale=0.225]{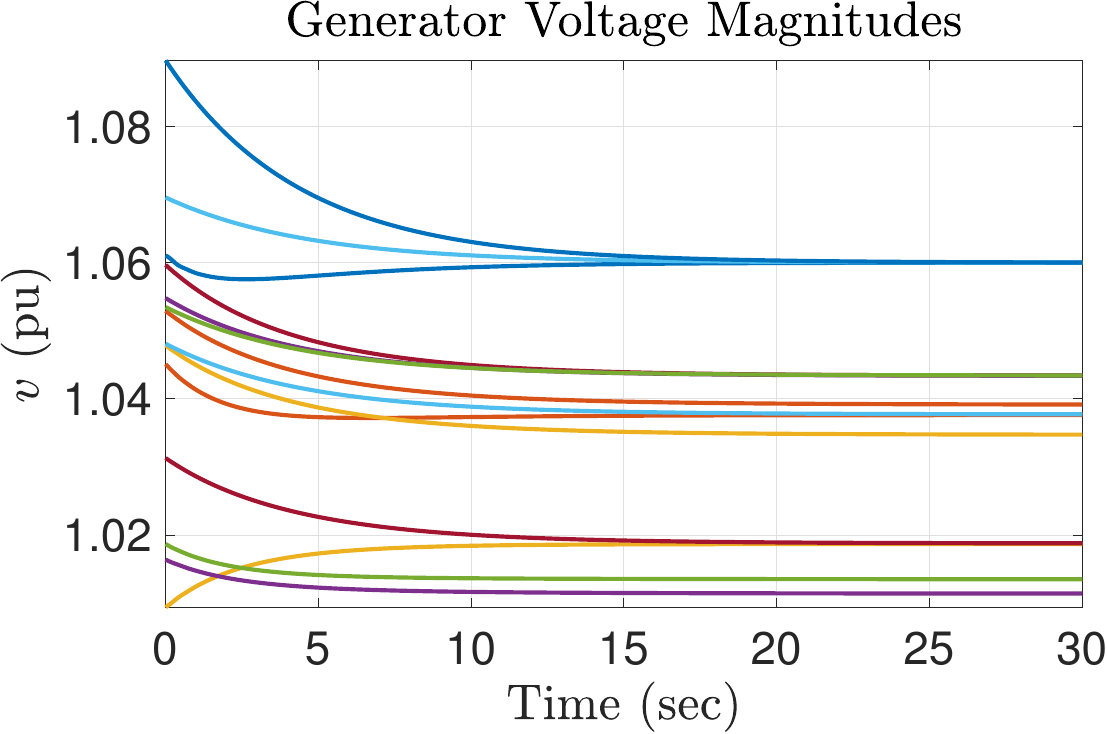}\label{<fi2>}}
	
	\subfloat[c][]{\includegraphics[scale=0.225]{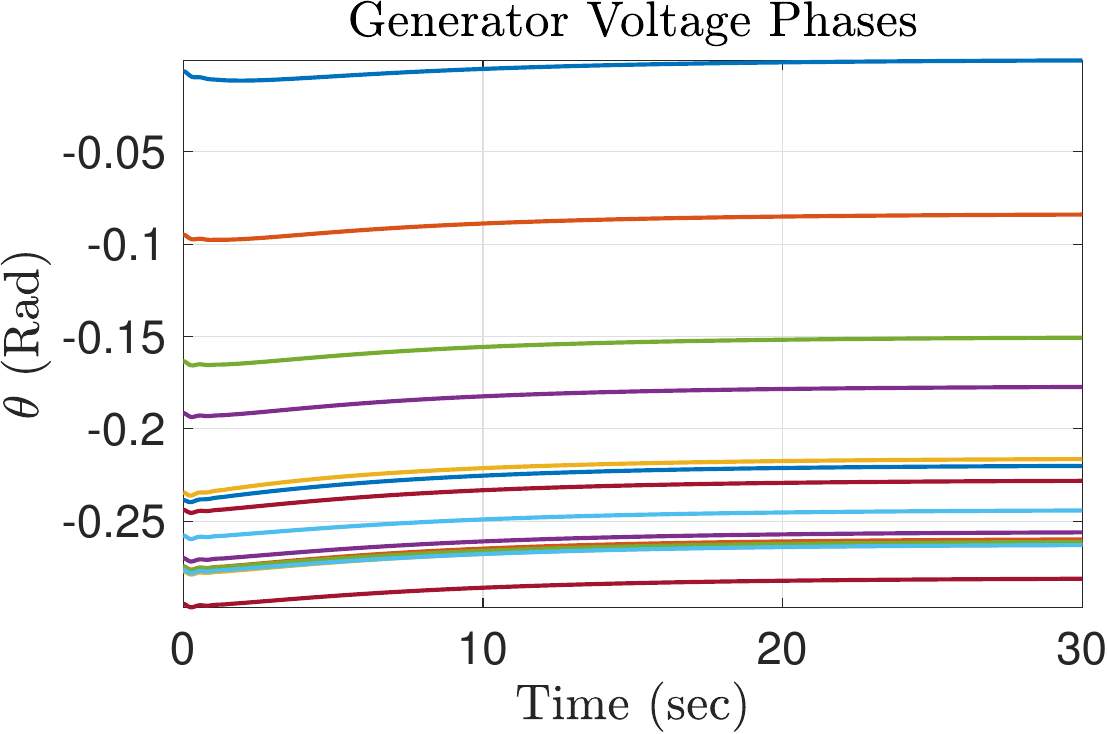}\label{<fi3>}}
	\subfloat[d][]{\includegraphics[scale=0.225]{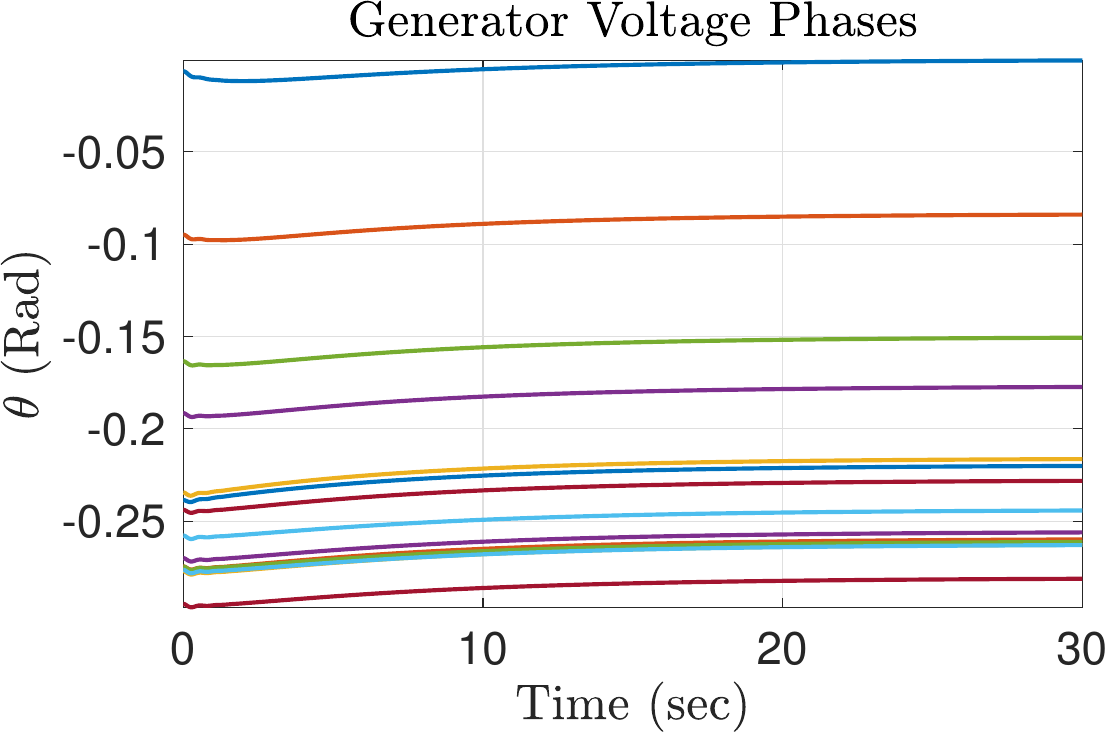}\label{<fi4>}}
	
	\subfloat[e][]{\includegraphics[scale=0.225]{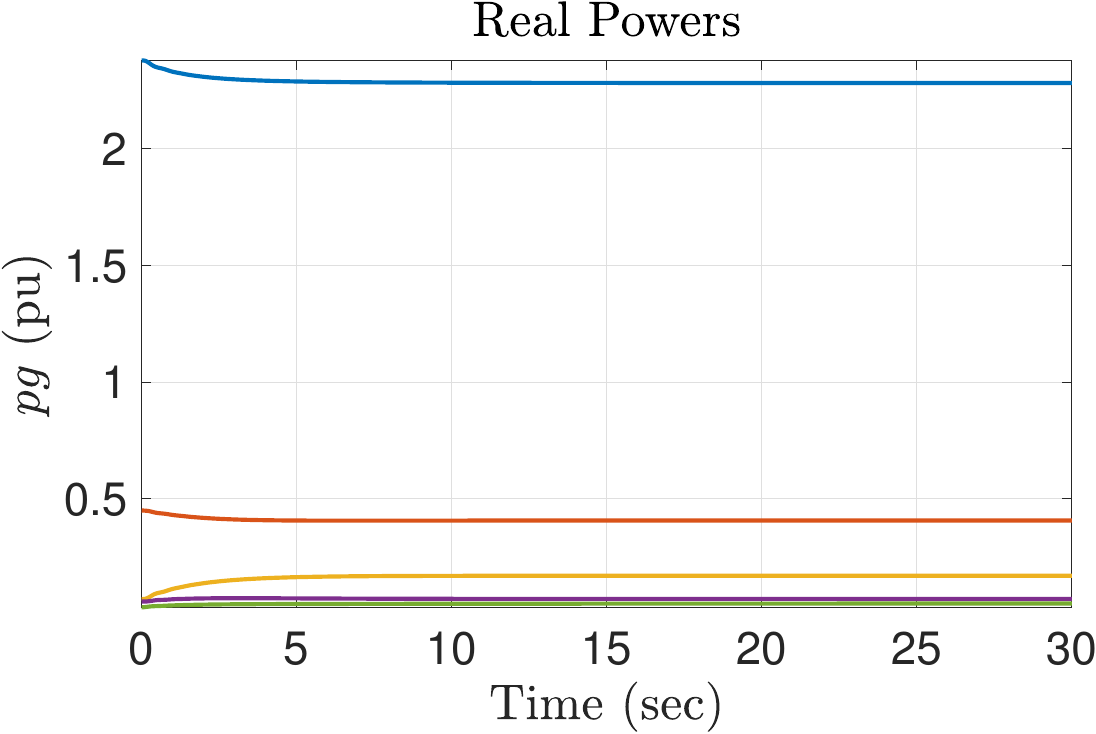}\label{<fi5>}}
	\subfloat[f][]{\includegraphics[scale=0.225]{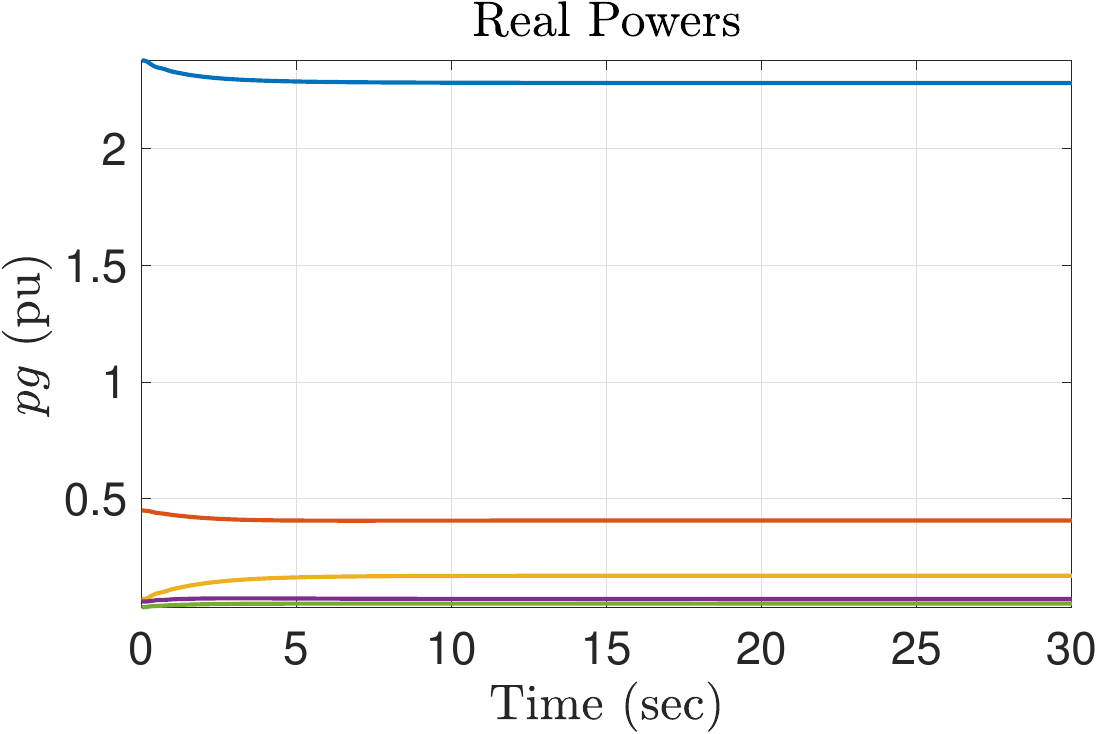}\label{<fi6>}}
	
	\subfloat[g][]{\includegraphics[scale=0.225]{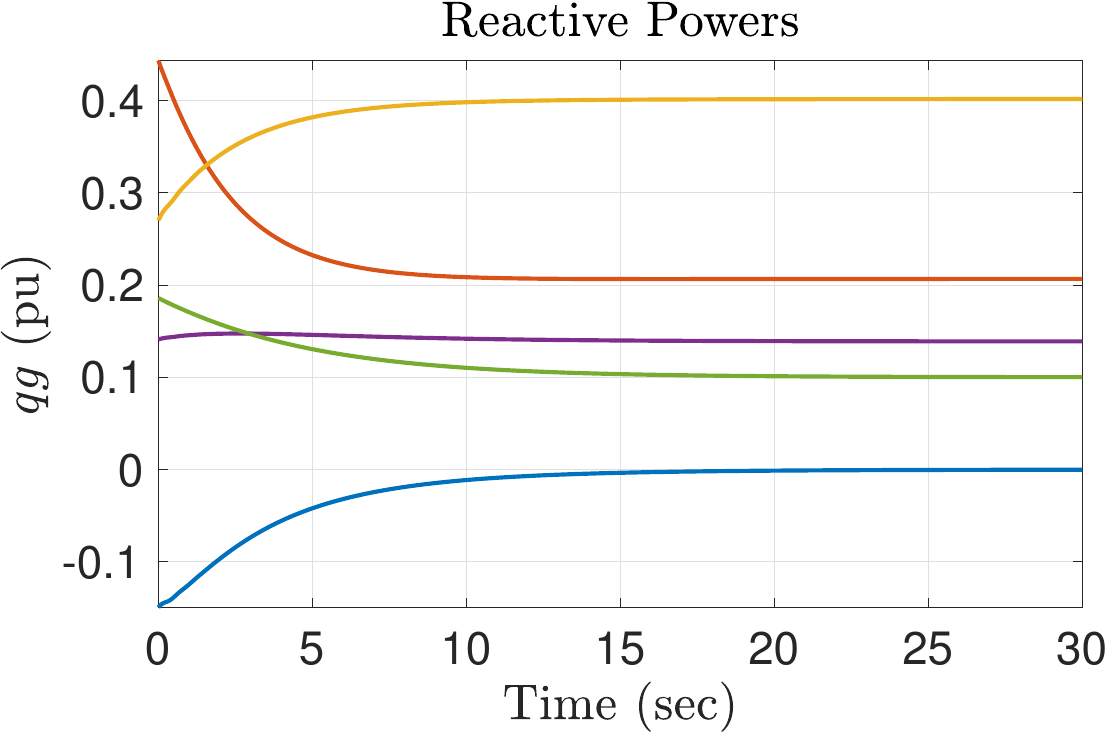}\label{<fi7>}}
	\subfloat[h][]{\includegraphics[scale=0.225]{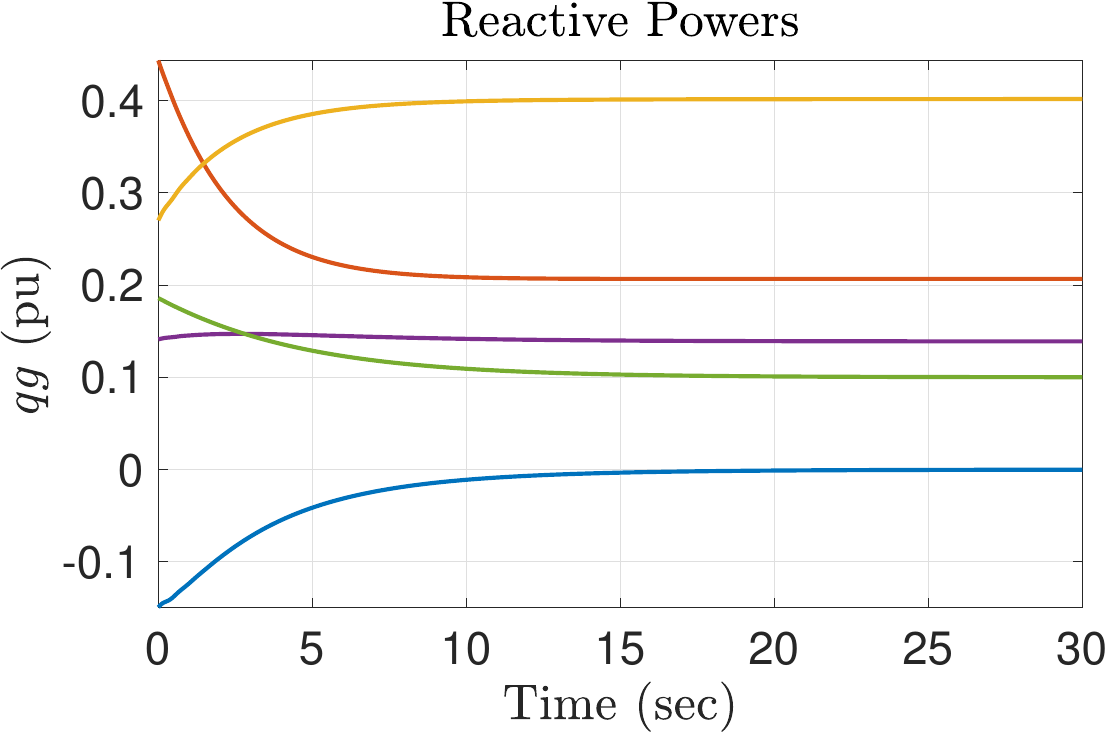}\label{<fi8>}}
	
	\caption{Some of the state trajectories---(a,b) generator voltage magnitudes, (c,d) generator voltage phases, (e,f) real powers, and (g,h) reactive powers---associated with $A+\Delta + BF^{\text{updated}}C$ and $A+\Delta + BF^{\text{typical}}C$ considering case $14$-bus, respectively.}
	\label{fig:stats}
\end{figure}

\section{Concluding Remarks} \label{Con}

This paper proposes a simple yet efficient update of a nominal stabilizing SOF controller. According to the derived theoretical and numerical results throughout the paper, we present the following answer to the question stated in Section \ref{sec:ProFor} (\textit{Q1}):

\noindent \textit{A1}: A least-squares problem built upon the notion of MDRP enables us to propose an efficient updated stabilizing SOF controller. For a known norm-bounded perturbation, we derive sufficient stability conditions followed by the characterized guaranteed stability regions. Geometric metrics are defined to quantify the distance-to-instability of the proposed updated stabilizing SOF controllers.

Moreover, our extensive numerical simulations corroborate that although we utilize a heuristic optimization method to compute the MDRP, it performs quite well in practice compared to an existing approximation method in the literature, namely the hybrid expansion-contraction (HEC) method.

\noindent \textit{Limitations:} Like any engineering solution, the proposed updated stabilizing SOF controller has some limitations. The main limitations are three-fold: (\textit{i}) we propose a semi-dynamic solution to a dynamic problem. The static nature comes from the utilized least-squares problem. The dynamic nature comes from the information stored in the nominal stabilizing SOF controller $F^{\text{nominal}}$ for the state-space triplet $(A,B,C)$ (i.e., $\beta_{\mathbb{R}}(A+BF^{\text{nominal}}C)$), (\textit{ii}) computing the exact value of MDRP $\beta_{\mathbb{R}}(A+BF^{\text{nominal}}C)$ is computationally intractable (NP-hard) and the practical heuristics to estimate $\beta_{\mathbb{R}}(A+BF^{\text{nominal}}C)$ may provide the less accurate values. The less accurate $\beta_{\mathbb{R}}(A+BF^{\text{nominal}}C)$, the less accurate guaranteed stability we get for the proposed update. Also, the more time-consuming practical heuristics we utilize to estimate $\beta_{\mathbb{R}}(A+BF^{\text{nominal}}C)$, the less efficient update we get, and (\textit{iii}) Unlike the typical update, the proposed update can be destabilizing for a subset of perturbations as illustrated by the region outside the guaranteed stability region $S_{\kappa}$ for $\kappa < 1$, i.e., $\beta_{\mathbb{R}}(A+BF^{\text{nominal}}C) < \rho$. However, the positive point about the proposed update is that, unlike the typical update, it always provides a non-empty guaranteed stability region (the typical approach can fail to propose an updated stabilizing SOF controller, as it is a complex problem in general).

Future work will focus on addressing some of these issues, but also extending this algorithm to update solutions for any control-driven SDP when the problem data changes. In particular, an architecture that finds a way to recompute feedback gains (or observer gains for state estimation) without the need to solve the often intractable large-scale SDPs (when quick updates are needed) can address some of the issues of SDP-based solutions. This paper focused only on the SOF architecture, but extensions to other control problems are interesting future directions. 

\bibliographystyle{IEEEtran}
\bibliography{References}

\begin{appendices}

\section{Proof of Lemma \ref{Lem1}} \label{ApndxA}

Considering $\Sigma_H = \begin{bmatrix} \Gamma_H^\top & 0 \end{bmatrix}^\top$ and noting that $U_H^\top U_H = I_{n^2}$ and $V_H^\top V_H = I_{mp}$ hold, we have
\begin{align*}
& P := I_{n^2} - H H^{+} = I_{n^2} - H (H^\top H)^{-1} H^\top =\\
& I_{n^2} - U_H \Sigma_H V_H^\top (V_H \Sigma_H^\top U_H^\top U_H \Sigma_H V_H^\top)^{-1} V_H \Sigma_H^\top U_H^\top=\\
& I_{n^2} - U_H \Sigma_H V_H^\top (V_H (\Gamma_H^{2})^{-1} V_H^\top) V_H \Sigma_H^\top U_H^\top =\\
& I_{n^2} - U_H \begin{bmatrix} \Gamma_H^\top & 0 \end{bmatrix}^\top (\Gamma_H^{2})^{-1} \begin{bmatrix} \Gamma_H^\top & 0 \end{bmatrix} U_H^\top=\\
& U_H I_{n^2} U_H^\top - U_H  \begin{bmatrix} I_{mp} & 0\\0 & 0 \end{bmatrix} U_H^\top= U_H \begin{bmatrix} 0 & 0\\0 & I_{n^2 - mp} \end{bmatrix} U_H^\top.
\end{align*}
Moreover, according to $H := C^\top \otimes B$ and the properties of Kronecker product, we get
\begin{align*}
H :=&~ C^\top \otimes B = (V_C \otimes U_B)(\Sigma_C^\top \otimes \Sigma_B)(U_C \otimes V_B)^\top =\\
&~(V_C \otimes U_B) (U_{\Omega} \Sigma_{\Omega} V_{\Omega}^\top) (U_C \otimes V_B)^\top =\\
&~((V_C \otimes U_B)U_{\Omega}) \Sigma_{\Omega} ((U_C \otimes V_B)V_{\Omega})^\top.
\end{align*}
Then, we have
\begin{align*} 
(U_H,\Sigma_H,V_H) &= ((V_C \otimes U_B)U_{\Omega},\Sigma_{\Omega},(U_C \otimes V_B)V_{\Omega})
\end{align*}
which completes the proof.

\section{Proof of Proposition \ref{Prop1}} \label{ApndxB}

Substituting \eqref{Pform} of Lemma \ref{Lem1} in \eqref{ROV}, we get
\begin{align*}
J^{\ast}(\Delta) =&~ \mathbf{vec}(\Delta)^\top U_H \begin{bmatrix} 0 & 0\\0 & I_{n^2 - mp} \end{bmatrix} U_H^\top \mathbf{vec}(\Delta) =\\ 
&~\delta^\top U_H \begin{bmatrix} 0 & 0\\0 & I_{n^2 - mp} \end{bmatrix} U_H^\top \delta.
\end{align*}
Then, defining $\chi := U_H^\top \delta$ and noting that $U_H U_H^\top = I_{n^2}$ holds (because $U_H$ is a unitary matrix), we get $\delta = U_H \chi$. Since $\delta^\top \delta = \chi^\top U_H^\top U_H \chi$, $U_H^\top U_H = I_{n^2}$, and $\delta^\top \delta = \|\Delta\|_F^2 = r^2$ hold, we get $\chi^\top \chi = r^2$ that inspires us to define $\psi := \frac{\chi}{\| \chi \|} = \frac{\chi}{r}$. Note that $\psi \in \mathbb{R}^{n^2}$ and $\|\psi\| = 1$ hold. Then, we have $\chi = r \psi$ and subsequently we get $\delta = U_H \chi = r U_H \psi$. Defining $\mu \in \mathbb{R}^{mp}$ and $\nu \in \mathbb{R}^{n^2-mp}$ as follows: $\mu := \begin{bmatrix} I_{mp} & 0 \end{bmatrix} \psi$, $\nu := \begin{bmatrix} 0 & I_{n^2-mp} \end{bmatrix} \psi$, we get $\psi = \begin{bmatrix} \mu^\top & \nu^\top \end{bmatrix}^\top$. Since $\|\psi\|^2 = \|\mu\|^2 + \|\nu\|^2 = 1$ holds, we can consider $\|\mu\| = \cos(\frac{\pi \theta}{2})$ and $\|\nu\| = \sin(\frac{\pi \theta}{2})$ for a $\theta \in [0,1]$. Then, we have
\begin{align*}
\mu & = \frac{\mu}{\|\mu\|} \cos \Big(\frac{\pi \theta}{2}\Big), \nu = \frac{\nu}{\|\nu\|} \sin\Big(\frac{\pi \theta}{2}\Big).
\end{align*}
Defining $\phi_c := \frac{\mu}{\|\mu\|}$ and $\phi_s := \frac{\nu}{\|\nu\|}$, we get $\mu = \phi_c \cos(\frac{\pi \theta}{2})$ and $\nu = \phi_s \sin(\frac{\pi \theta}{2})$ (Note that $\|\phi_c\| = 1$ and $\|\phi_s\|=1$ hold). Then, considering the $r = \rho \sin (\frac{\pi \tau}{2})$ with $\tau \in ]0,1]$, we have
\begin{align*}
\Delta =&~\mathbf{vec}^{-1}(\delta) = \mathbf{vec}^{-1}(r U_H \psi) =\\ 
&~r \mathbf{vec}^{-1}\bigg((V_C \otimes U_B)U_{\Omega}\begin{bmatrix} \phi_c \cos(\frac{\pi \theta}{2}) \\ \phi_s \sin(\frac{\pi \theta}{2}) \end{bmatrix}\bigg) = \\
&~\rho \sin \Big(\frac{\pi \tau}{2}\Big) U_B \mathbf{vec}^{-1}\bigg(U_{\Omega}\begin{bmatrix} \phi_c \cos(\frac{\pi \theta}{2}) \\ \phi_s \sin(\frac{\pi \theta}{2}) \end{bmatrix}\bigg) V_C^\top
\end{align*}
which completes the proof of \eqref{MR}.

Also, for $J^{\ast}(\Delta)$ in \eqref{ROV}, we have
\begin{align*}
& J^{\ast}(\Delta) = \mathbf{vec}(\Delta)^\top P \mathbf{vec}(\Delta) = \delta^\top P \delta \overset{(\delta = r U_H \psi)}{=}\\
& r^2 \psi^\top U_H^\top P U_H \psi = r^2 \begin{bmatrix} \mu^\top & \nu^\top \end{bmatrix} \begin{bmatrix} 0 & 0\\0 & I_{n^2 - mp} \end{bmatrix} \begin{bmatrix} \mu \\ \nu \end{bmatrix} = \\ 
& r^2 \nu^\top \nu = r^2 \Big(\phi_s \sin\Big(\frac{\pi \theta}{2}\Big)\Big)^\top \phi_s \sin\Big(\frac{\pi \theta}{2}\Big) =\\
& \Big(r \sin \Big(\frac{\pi \theta}{2}\Big) \Big)^2 \|\phi_s\|^2 = \Big(\rho \sin \Big(\frac{\pi \tau}{2}\Big) \sin \Big(\frac{\pi \theta}{2}\Big) \Big)^2
\end{align*}
which completes the proof of \eqref{JMR}.

\section{Proof of Proposition \ref{Prop2}} \label{ApndxC}

We use \eqref{ImpIneq} as a sufficient condition on the stability of the updated perturbed state-space \eqref{RDPerturbedSD}. By substituting $G^{\ast}_{\Delta}$ in \eqref{ImpIneq}, we get
\begin{align} \label{ssbr}
\sin \Big(\frac{\pi \tau}{2} \Big) \sin \Big(\frac{\pi \theta}{2} \Big) < \frac{\beta_{\mathbb{R}}(A+BFC)}{\rho}.
\end{align}

If $\rho < \beta_{\mathbb{R}}(A+BFC)$ holds, then $F + G^{\ast}_{\Delta}$ with $G^{\ast}_{\Delta}$ in \eqref{Osol} is an updated stabilizing SOF controller because the left-hand-side of \eqref{ssbr} is at most $1$ and the right-hand-side of \eqref{ssbr} is greater than $1$. Then, \eqref{ssbr} holds.

If $\rho \ge \beta_{\mathbb{R}}(A+BFC)$ holds, since $\sin (\frac{\pi \theta}{2})$ attains its maximum value at $\theta = 1$, \eqref{ssbr} reduces to
\begin{align*}
\sin \Big(\frac{\pi \tau}{2} \Big) < \frac{\beta_{\mathbb{R}}(A+BFC)}{\rho}
\end{align*}
or equivalently
\begin{align*}
\tau <  \frac{2}{\pi} \arcsin\Big(\frac{\beta_{\mathbb{R}}(A+BFC)}{\rho}\Big)
\end{align*} 
from which, we define $\kappa$ in \eqref{SSS}. Similarly, we may extract the definition of $\zeta_{\tau,\kappa}$ in \eqref{SSS}. Thus, if $(\tau_{\Delta},\theta_{\Delta}) \in S_{\kappa}$ holds, then $F + G^{\ast}_{\Delta}$ with $G^{\ast}_{\Delta}$ in \eqref{Osol} is an updated stabilizing SOF controller. 

The expression in \eqref{AreaM} expresses the area of $S_{\kappa}$ divided by the area of unit square $]0,1] \times [0,1]$ in $2$-dimensional parametric space of $(\tau,\theta)$. Note that $\int_{0}^{\kappa} 1 d \tau = \kappa$ has simplified the right-hand-side of \eqref{AreaM}. To show that $\xi_{\kappa}$ is an increasing function of $\kappa$, we compute the derivative of $\xi_{\kappa}$ with respect to $\kappa$ as follows (utilizing the Leibniz integral rule \cite{apostol1991calculus}):
\begin{align} \label{Der}
\frac{d \xi_{\kappa}}{d\kappa} &= \cos \Big (\frac{\pi \kappa}{2} \Big ) \int_{\kappa}^{1} \frac{1}{\sqrt{\sin (\frac{\pi \tau}{2})^2 - \sin (\frac{\pi \kappa}{2})^2}} d \tau.
\end{align}
According to \eqref{Der}, $\frac{d \xi_{\kappa}}{d\kappa} \ge 0$ holds and noting that  
\begin{align*}
\frac{d \kappa}{d \rho} &= -\frac{2 \beta}{\pi \rho \sqrt{\rho^2 - \beta^2}} < 0, \frac{d \kappa}{d \beta} = \frac{2}{\pi \sqrt{\rho^2 - \beta^2}} > 0 
\end{align*}
hold, the proof is complete.

\end{appendices}

\begin{IEEEbiography}
[{\includegraphics[width=1in,height=1.25in,clip,keepaspectratio]{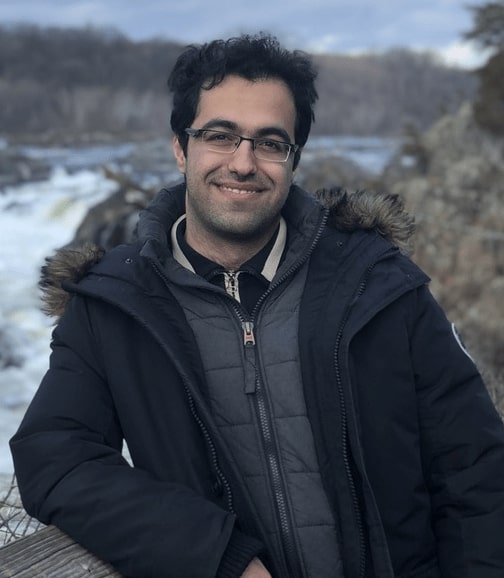}}] {MirSaleh Bahavarnia} \; (Member, IEEE) received the B.Sc. degree in electrical engineering (control) and a certificate of the minor program in mathematics from the Sharif University of Technology, Tehran, Tehran, Iran, in 2013 and the Ph.D. degree in mechanical engineering (control) from the Lehigh University, Bethlehem, PA, USA, in 2018. 

From 2018 to 2020, he was a Postdoctoral Research Associate with the Department of Electrical and Computer Engineering and the Institute for Systems Research (ISR), University of Maryland, College Park, MD, USA. Since 2022, he has been a Postdoctoral Research Scholar with the Department of Civil and Environmental Engineering, Vanderbilt University, Nashville, TN, USA. His research interests include distributed control, feedback control, power systems control, process control, robust control, and traffic control.
\end{IEEEbiography}
 \begin{IEEEbiography}[{\includegraphics[width=1in,height=1.25in,clip,keepaspectratio]{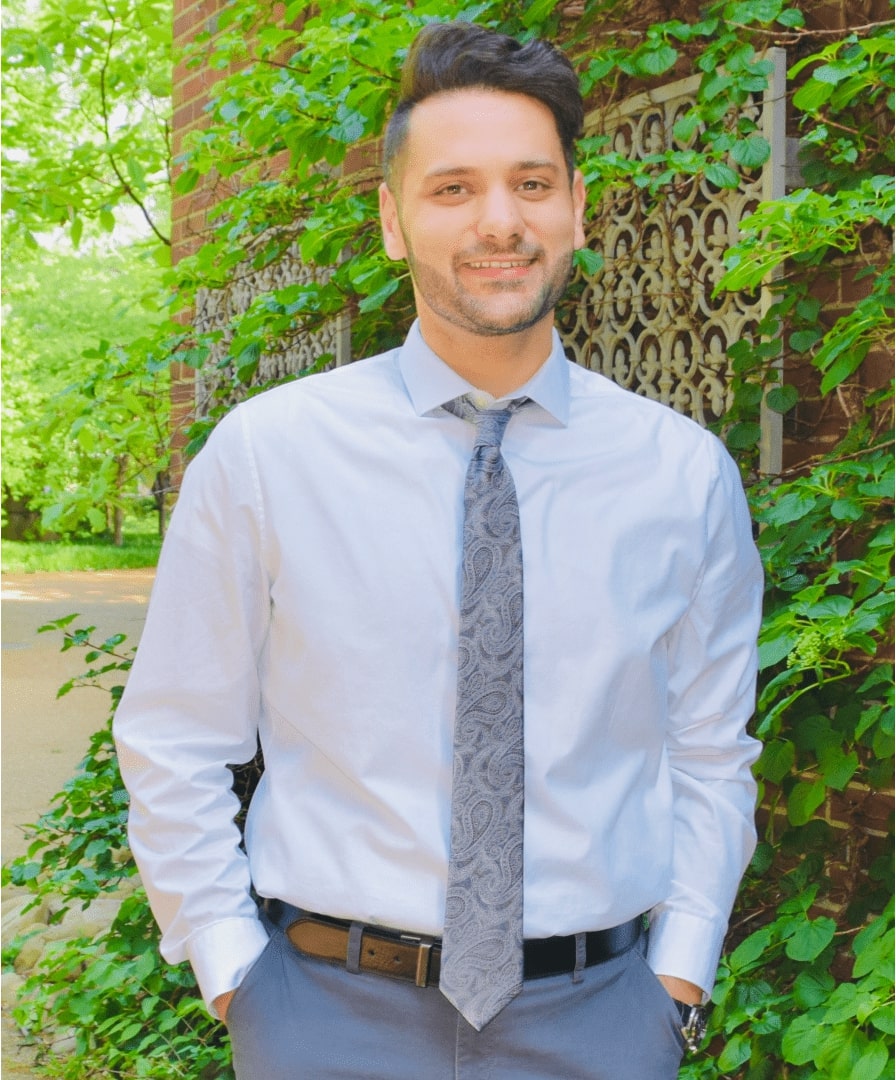}}] {Ahmad F. Taha} \; (Member, IEEE) received the B.E. degree in electrical and computer engineering from the American University of Beirut, Beirut, Lebanon, in 2011, and the Ph.D. degree in electrical and computer engineering from Purdue University, West Lafayette, IN, USA, in 2015. 
 
 Before joining Vanderbilt University, Nashville, TN, USA, he was an Assistant Professor with the ECE Department at the University of Texas, San Antonio. He is an Associate Professor with the Department of Civil and Environmental Engineering at Vanderbilt University. He has a secondary appointment in the ECE Department. His research interests include understanding how complex cyber-physical and urban infrastructures operate, behave, and occasionally misbehave, and optimization, control, monitoring, and security of infrastructure with power, water, and transportation systems applications.
 
 Dr. Taha is an Associate Editor for IEEE Transactions on Control of Network Systems.

\end{IEEEbiography}

\end{document}